\documentclass[10pt,journal,compsoc]{IEEEtran}
\IEEEoverridecommandlockouts
\usepackage{cite}
\usepackage{amsmath}
\usepackage{amsmath,amssymb,amsfonts}
\usepackage[linesnumbered,ruled,vlined]{algorithm2e}
\usepackage{graphicx}
\usepackage{textcomp}
\usepackage{xcolor}
\usepackage{floatrow}
\usepackage{subfigure}
\usepackage{amsthm}
\usepackage{mdwlist}
\usepackage{changepage}
\usepackage{makecell,rotating,multirow,diagbox}
\usepackage{epstopdf}
\usepackage{color}
\usepackage{cite}
\usepackage{bm}
\usepackage{url}
\usepackage{threeparttable}
\usepackage[usestackEOL]{stackengine}
\usepackage{pifont}
\usepackage{booktabs}
\usepackage{framed}
\usepackage{mdframed}
\usepackage{multicol}
\usepackage{multirow}
\usepackage{tikz}
\usepackage{arydshln} %虚线
\usepackage{utfsym} %对勾
\def\BibTeX{{\rm B\kern-.05em{\sc i\kern-.025em b}\kern-.08em
		T\kern-.1667em\lower.7ex\hbox{E}\kern-.125emX}}
\begin{document}
	
	\title{Cross-Channel: Scalable Off-Chain Channels Supporting Fair and Atomic Cross-Chain Operations\\
		%Yiao: 题目包含三个关键点，一个是高效率：在调用相同次数智能合约前提下可以处理更多的交易；其二是更通用：支持物物交换；其三是支持跨链
	}
	
	\author{Yihao Guo, Minghui Xu, Dongxiao Yu, Yong Yu, Rajiv Ranjan, Xiuzhen Cheng
		\thanks{Corresponding author: Minghui Xu.}
		\thanks{Y. Guo, M. Xu, D. Yu, and X. Cheng are with the School of Computer Science and Technology, Shandong University, Qingdao, Shandong, China (e-mail:\{ yhguo@mail., mhxu@, dxyu@, xzcheng@\}sdu.edu.cn).}
		\thanks{Y. Yu is with the School of Computer Science, Shaanxi Normal University, Xi’an, China (e-mail: yuyong@snnu.edu.cn).}
 	\thanks{R. Ranjan is with the School of Computing, Newcastle University, Newcastle, United Kingdom (e-mail: raj.ranjan@newcastle.ac.uk).}
	}

	\IEEEtitleabstractindextext{
	\begin{abstract}
		
		% Cross-chain and off-chain technologies are critical solutions to scaling blockchains. 
		
		Cross-chain technology facilitates the interoperability among isolated blockchains on which users can freely communicate and transfer values. Existing cross-chain protocols suffer from the scalability problem when processing on-chain transactions. Off-chain channels, as a promising blockchain scaling technique, can enable micro-payment transactions without involving on-chain transaction settlement. However, existing channel schemes can only be applied to operations within a single blockchain, failing to support cross-chain services. Therefore in this paper, we propose $\mathsf{Cross}$-$\mathsf{Channel}$, the first off-chain channel to support cross-chain services. We introduce a novel hierarchical channel structure, a new hierarchical settlement protocol, and a smart general fair exchange protocol, to ensure scalability, fairness, and atomicity of cross-chain interactions. Besides, $\mathsf{Cross}$-$\mathsf{Channel}$ provides strong security and practicality by avoiding high latency in asynchronous networks.Through a 50-instance deployment of $\mathsf{Cross}$-$\mathsf{Channel}$ on AliCloud, we demonstrate that $\mathsf{Cross}$-$\mathsf{Channel}$ is well-suited for processing cross-chain transactions in high-frequency and large-scale, and brings a significantly enhanced throughput with a  small amount of gas and delay overhead.
	\end{abstract}
	
	}
	\maketitle

\section{Introduction}\label{sec:introduction}
%区块链的优势
%Blockchain technologies have attracted tremendous interests in recent years from government and academia to industry for their properties of decentralization, transparency, and immutability~\cite{belchior2021survey}. 
%They can be employed to solve the traditional single point of failure problem, achieve effective access control, or even provide a zero-trust fault tolerant environment enabling mutually distrustful parties to establish a trust relationship in wireless networks~\cite{lao2020survey, xu2022blown, xu2021wchain}. Hence, blockchain is quickly introduced to many areas including finance, smart grid, cloud services, and IoV (Internet of Vehicle)~\cite{nakamoto2008bitcoin, guo2022blockchain, 9837445, xu2022cloudchain,zhou2019blockchain}.
%%区块链的问题-孤岛->解决方案-跨链->跨链问题-效率
%However, in addition to the aforementioned great benefits, blockchain-based systems suffer from the isolation problem, which prevents information communications, value exchanges, and collaborative operations between blockchains, hindering the advantages of blockchain technologies in consensus, trust, and co-operations.  
%%low transaction rates and high transaction processing latencies, which hinder blockchains’ scalability.

It is well-known that blockchain-based systems suffer from the  information isolation problem~\cite{xie2019survey}, which prevents %information communications, 
message transports, value exchanges, and collaborative operations among blockchains, hindering the advantages of blockchain technologies in consensus, trust, and cooperations~\cite{xu2022spdl,chen2022survey, xu2022cloudchain,guo2022blockchain,liu2022blockchain}. Cross-chain technology has been considered to be one of the effective ways to solve the isolation problem, aiming to build a bridge for the communications and coordinations among isolated blockchain systems~\cite{belchior2021survey}. 
Nevertheless, most existing cross-chain techniques such as notaries in notary schemes~\cite{belchior2021survey} and relay chains in sidechains/relays~\cite{wood2016polkadot}, rely on third parties that are assumed to be safe,  which reduces their availability and makes them extremely vulnerable to the single point of failure problem.  
Hashed TimeLock Contract (HTLC)~\cite{poon2016bitcoin} is a decentralized cross-chain scheme that employs smart contracts to ensure atomicity of transactions. But unfortunately, in HTLC, one cross-chain exchange requires multiple on-chain consensus, which would undoubtedly decrease the transaction rate and increase the transaction latency, further deepening the blockchain system's poor scalability. %, hindering blockchains’ scalability.
Therefore, one can conclude that current cross-chain schemes cannot realize efficient cross-chain interoperability without relying on third parties.

%This problem is mainly caused by the involved single chains themselves.
%More specifically, the theoretical throughput of the Bitcoin~\cite{nakamoto2008bitcoin} reaches only 7 transactions per second (TPS) but in practice, the average throughput is just 3 because of issues such as empty blocks. Ethereum~\cite{wood2014ethereum} suffers from the same problem: its transaction throughput is about 15 TPS but is typically dropped to a much lower value when interactions with smart contracts are needed. 
%Furthermore, as a platform becomes more and more popular, its transaction load gets higher and higher, making the network congestion reach a new, previously unseen level~\cite{Ethworks}. 
%%So, we can get that due to the poor scalability of blockchain itself, the development of cross-chain services are limited.
%Therefore, one can conclude that blockchain system's poor scalability can be significantly deepened after considering cross-chain services and current cross-chain schemes cannot realize efficient cross-chain interoperability without relying on third parties.

%单链-效率由通道解决->跨链空白->跨链通道
In fact, researchers have put forward  feasible schemes such as blockchain sharding, off-chain channels, and roll-up techniques (zk-rollup, optimistic rollup)~\cite{qi2022s, huang2022brokerchain, papadis2022payment, xie2019survey}, to enhance the scalability of single blockchain systems. 
Among them, off-chain channels (see an example shown in Fig.~\ref{Fig:UAS}), which employ on-chain processes to establish and close a channel and off-chain operations to carry out tasks within the channel, provide faster transaction processing, need lower effort in hardware configuration, and have been successfully applied to Lightning Network~\cite{poon2016bitcoin}. 
Therefore, it is of great significance to  extend the current off-chain channel schemes for cross-chain services. Nevertheless, this is a nontrivial task. There exist two open challenges  that should be addressed in order to take advantage of the high throughput of channels for cross-chain operations. 
%To realize this goal, we consider designing a cross-chain channel scheme named $\mathsf{Cross\text{-}Channel}$. %As far as we know has studied, $\mathsf{Cross\text{-}Channel}$ is the first channel scheme that supports cross-chain services with high throughput. 
%The design of $\mathsf{Cross\text{-}Channel}$ faces the following challenges:
%As for the second problem, the isolation of the network hinders the collaborative operation between different blockchains. Cross-chain technology is an important method to solve blockchain interconnection, aiming to build a bridge of communication and coordination for isolated blockchain systems. However, most current cross-chain schemes rely on the security of third parties, such as notaries in the Notary scheme~\cite{belchior2021survey}, relay chains in the Side-chain/Relay scheme~\cite{kwon2019cosmos}, which are highly vulnerable to the single point of failure~\cite{lynch2009single}. Although the Hashed Timelock Contract (HTLC) scheme~\cite{poon2016bitcoin} does not depend on a third party, one cross-chain exchange requires multiple on-chain consensus, making it unsuitable for the large-scale cross-chain transactions situation. 

%h 表示当前位置：将图形放置在正文文本中给出该图形环境的地方。如果本页所剩的页面不够，这一参数将不起作用
%t 表示顶部：将图形放置在页面的顶部。
%b 表示底部：将图形放置在页面的底部。
%p 表示浮动页：将图形放置在一只允许有浮动对象的页面上。
%跨链给通道带来的挑战
\begin{figure}[t!]
	\centering
	\includegraphics[width=0.8\textwidth]{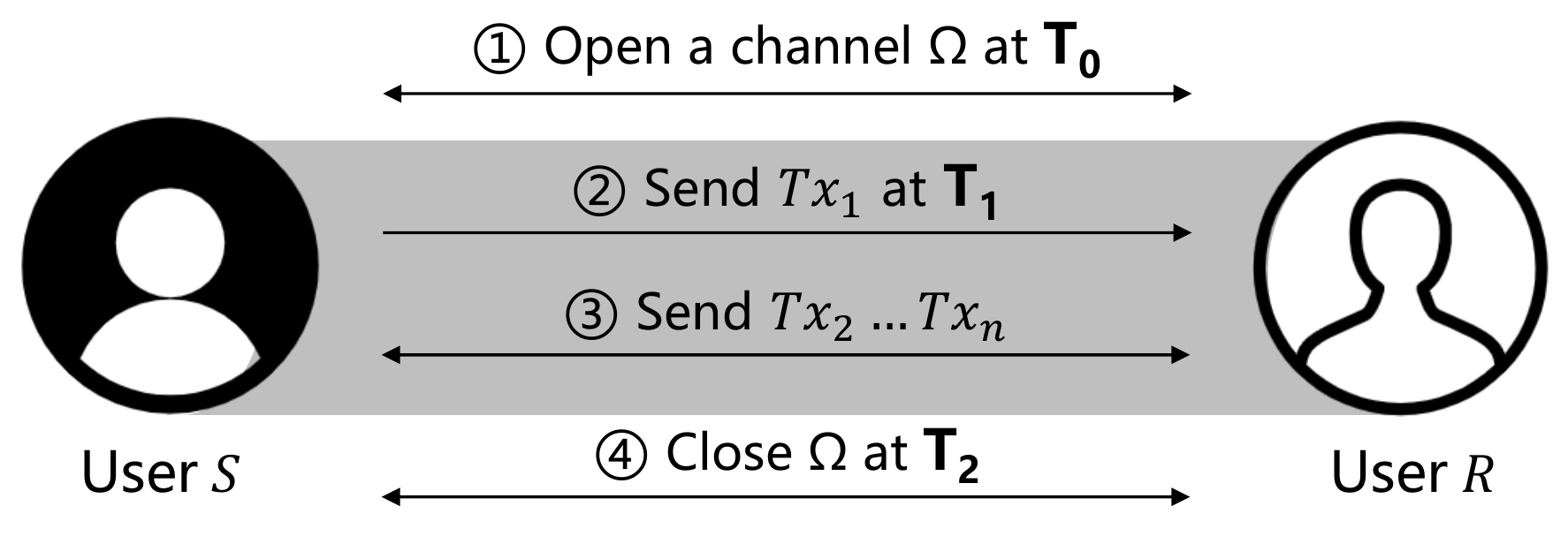}
	\caption{A channel scheme example, in which a solid line represents an interaction process, and the shaded belt represents a payment channel, a state channel or a virtual channel. Note that step \textcircled{1} and step \textcircled{4} are on-chain processes that open and close the channel $\Omega$, respectively. In step \textcircled{2}, $\mathcal{S}$ sends transaction $Tx_1$ to $\mathcal{R}$. Step \textcircled{3} indicates that $\mathcal{S}$ and $\mathcal{R}$ send more transactions to each other via this channel. These transaction operations take place in the channel $\Omega$ and do not consume any on-chain resource.
	}
	\label{Fig:UAS}
\end{figure}
%But unfortunately, although the channel scheme has its own advantages, it has some remaining challenges. 
First, current channel schemes such as payment channels~\cite{poon2016bitcoin, hearn2013micro, decker2015fast, papadis2022payment}, state channels~\cite{dziembowski2018general,coleman2018counterfactual} and virtual channels~\cite{dziembowski2019perun, dziembowski2019multi} cannot support spending unsettled amounts. As shown in Fig.~\ref{Fig:UAS}, suppose sender $\mathcal{S}$ sends some amount $x$ in $\mathsf{Tx_1}$ to receiver $\mathcal{R}$ at time $\mathsf{T_1}$. $\mathcal{R}$ cannot use $x$ before both parties successfully close the channel at $\mathsf{T_2}$. In fact, the longer time the channel stays open, the longer the time ($\mathsf{T_2}$-$\mathsf{T_1}$) $\mathcal{R}$ needs to wait before using $x$. This latency becomes even bigger  
%The difficulty in solving this problem lies in how to introduce other nodes to interact on the basis of the existing channel and ensure the correctness of the final settlement.
when cross-chain operations are involved as the multi-chain heterogeneous design brings more dimensions of complexity. % would cause even longer latency than intra-chain interactions. %, limiting the scalability of the system. 
We term this challenge \emph{Unsettled Amount Congestion (UAC) problem}. To address this issue, one needs to design a new channel architecture and a corresponding smart contract protocol to support flexible user joins while ensuring the correctness of settlement.

%. At the same time, the architecture also requires a special smart contract protocol to ensure the correctness of settlement.

%The key solution is to allow channel initiators to establish sub-channels with other participants, improve the throughput of transaction processed, and ensure the correctness of the final settlement.
%Therefore, supporting the unsettled amount of spending is conducive to further improve the efficiency of $\mathsf{Cross\text{-}Channel}$ and 
%We define the first challenge as the Unsettled Amount Congestion (UAC) problem.
%通道支支持简单的交换，没有考虑fair exchange 问题。
%Second, the current channel schemes are vulnerable to the incentive manipulation attack. Suppose user $\mathcal{A}$ and user $\mathcal{B}$ interact multiple times in the channel and $\mathcal{A}$ is a malicious adversary. Harris and Zohar show that $\mathcal{A}$ can delay $\mathcal{B}$’s transaction confirmation by overloading the system with his own transactions.
%Sup that a and B interact multiple times in the channel
%When $\mathcal{B}$ and $\mathcal{C}$ 

%TODO:可以从公平的角度出发
%The use of cryptocurrencies
%for payment not only exposes cloud service providers to the risk of sharp currency price fluctuations, but also is vulnerable to
%government bans due to regulatory concerns. So far payment via
%fiat money is still the norm in business and is naturally free of
%above troubles.
\begin{figure}[t]
	\centering
	\includegraphics[width=0.8\textwidth]{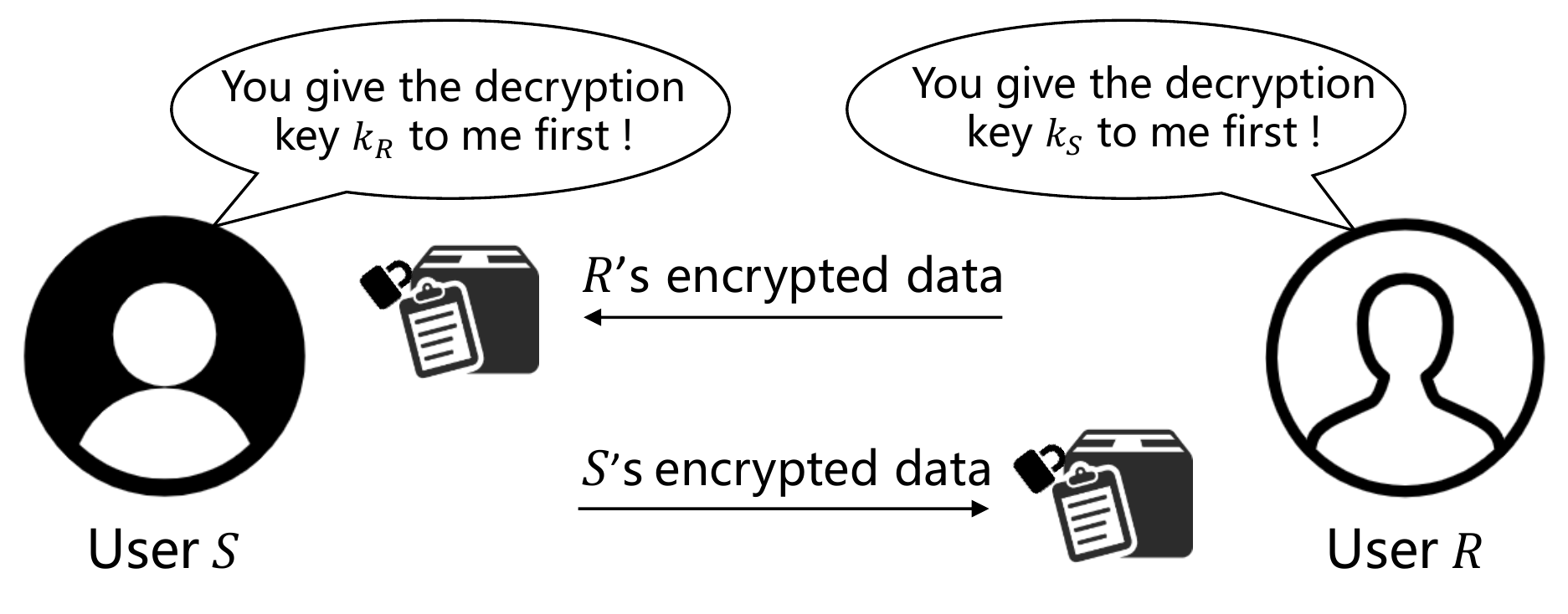}
	\caption{An example of the Unfair Exchange (UE) problem. A solid line represents an interaction process. }
	\label{Fig:UE}
\end{figure}

Second, the scenarios of cross-chain interactions are diversified, which requires a protocol to help channels ensure interaction fairness.
%, but the current channel scheme is not general enough to support various cross-chain services. In the current relevant research, in contrast to payment channel networks that only support off-chain payments between users, state channel networks~\cite{dziembowski2018general} allow execution of arbitrary complex smart contracts. But the state channel scheme can not guarantee the fairness of the interaction process, because it ignores that most of the information itself is valuable and must be encrypted. 
Current fair exchange schemes such as ZKCP~\cite{zkcp} and FairSwap~\cite{dziembowski2018fairswap} rely on cryptocurrencies for settlement, which assumes that one of the exchange objects can be made public (generally this object defaults to cryptocurrency).
%而在跨链场景中，有代币波动影响，更多的用户选择以物换物的交换方式，这使得现在的方案无法保证公平性。
%exposes users to the risk of sharp currency price fluctuations (more serious in multi-chain interaction). 
%The reason is that current fair exchange schemes assume that the transaction object of at least one of the two parties must be disclosed (generally this object defaults to currency) and they do not support the simultaneous disclosure of both parties' secrets. For example, current schemes allow a user $\mathcal{B}$ to exchange currency for currency or encrypted information with the other user $\mathcal{C}$. But when the objects exchanged by $\mathcal{B}$ and $\mathcal{C}$ cannot be disclosed, like encrypted information exchange, current schemes cannot guarantee the fairness of the interaction.
Unfortunately, in a cross-chain scenario, affected by the risk of sharp currency price fluctuations, more users choose to interact in a barter way. This makes current fair exchange schemes fail to guarantee the fairness of the interaction process, as none of the two parties might be willing to be the first to disclose its secret that is employed to protect the exchange object  (shown in Fig.~\ref{Fig:UE}). We define this to be the \emph{Unfair Exchange (UE) problem}, which severely limits the applications of cross-chain protocols. The difficulty in solving this problem lies in how to ensure the disclosure of both parties' secrets. 

In this paper, we propose $\mathsf{Cross}$-$\mathsf{Channel}$ to effectively address the above two challenges.
First, in order to solve the UAC problem, we design a novel hierarchical channel architecture with a hierarchical settlement protocol, which allows channel initiators to establish sub-channels with other participants in order to use an unsettled amount, improve the throughput of the processed transactions, and ensure the correctness of the final settlement.
Moreover, we present a general fair exchange protocol based on zk-SNARK~\cite{groth2016size} and $(t,n)$-VSS (Verifiable Secret Sharing)~\cite{pedersen1991non} to address the UE problem. Specifically, this protocol adopts zk-SNARK to guarantee the authenticity and privacy of information, and employs $(t,n)$-VSS with a smart contract to ensure the fair disclosure of both parties' secrets.
Finally, we adopt the HTLC to enhance the atomicity of cross-chain interactions. Note that HTLC was originally proposed for synchronous networks, while in an asynchronous network, some nodes may be affected by high latency and cannot successfully receive or send messages, which makes blocked nodes unable to complete the protocol due to timeouts. Therefore, we develop an incentive mechanism to make HTLC, thereby $\mathsf{Cross\text{-}Channel}$, suitable for asynchronous networks. % which while considering the impact of asynchronous networks in our design. In the asynchronous network, some nodes may be affected by high latency and cannot receive or successfully send messages, which makes blocked nodes unable to complete the protocol due to timeout. So, we provide  an incentive mechanism to help $\mathsf{Cross\text{-}Channel}$ be suitable for asynchronous networks.
%so that it can effectively enhance atomicity of cross-chain interactions. %In the asynchronous network, some nodes may be affected by high latency and cannot receive or successfully send messages.

For convenience, we highlight our contributions as follows:

\begin{enumerate}
	%propose present develop realize design
	%TODO: Minghui:    --03
	%想一下第一条跨链给通道带来的挑战
	%把第三条拆成两部分
	%把实验放到第三部分中
	%完成
	\item  $\mathsf{Cross}$-$\mathsf{Channel}$ is the first channel scheme to support cross-chain operations in both synchronous and asynchronous networks. To the best of our knowledge, current channel schemes only support intra-chain operations.% within synchronous networks. %, especially well-suited for processing large-scale transactions. 
	%and can be applied to asynchronous networks.
	
	\item We design a novel hierarchical channel architecture with a hierarchical settlement protocol, which can effectively solve the UAC problem. The proposed new architecture can further improve the throughput of $\mathsf{Cross}$-$\mathsf{Channel}$ and is well-suited for processing large-scale transactions.

	\item A general fair exchange protocol is proposed in this paper to guarantee the disclosure of both parties' secrets, ensure the fairness of cross-chain interactions, and further solve the UE problem.
	
	\item $\mathsf{Cross}$-$\mathsf{Channel}$ can support various cross-chain operations, especially for the exchange of encrypted information that does not rely on cryptocurrencies. %Based on We add an incentive protocol to the scheme to help our scheme adapt to high-latency asynchronous networks and effectively resist attacks. 
	
	\item Extensive simulation experiments in AliCloud are conducted to validate the performance of $\mathsf{Cross}$-$\mathsf{Channel}$. %It brings N times the throughput with a small amount of extra overhead on the chain.
\end{enumerate}

%文章组织
The rest of the paper is organized as follows.
We first review the most related work in Section~\ref{sec:related work}.
Then, we present the models of our scheme, %and design goals of our scheme, 
and briefly introduce the necessary preliminary knowledge in Section~\ref{sec:Preliminaries}.
% we propose a novel technique $\mathsf{zk\text{-}GSigproof}$ and provide a security analysis.
In Section~\ref{sec:Cross-Channel}, we detail our scheme $\mathsf{Cross}$-$\mathsf{Channel}$ considering different applications. %, including the overview, basic protocol details and deeper discussion.
%In Section~\ref{sec:analysis}, we analyze the security and privacy of $\mathsf{Vehicloak}$, and discuss some important issues in more details.
Section~\ref{sec:Implementation} reports the simulation experiments on AliCloud to evaluate the performance of $\mathsf{Cross}$-$\mathsf{Channel}$.
Finally, we provide concluding remarks in Section~\ref{sec:Conclusion}.
	\section{Related Work and Motivation}\label{sec:related work}
In this section, we introduce a few well-known off-chain channel schemes related to our design, including payment channels, state channels and virtual channels. 

\noindent\textbf{Payment channel schemes.} 
A payment channel is a temporary off-chain trading channel for improving the transaction throughput of the entire system. 
It was originally designed as a one-way channel~\cite{hearn2013micro}, and later evolved into a bi-directional channel so that one party can be both a sender and a receiver~\cite{decker2015fast}.
The most widely discussed recent projects are
Lightning Network~\cite{poon2016bitcoin} and Raiden~\cite{Raiden}, which establish payment channels in Bitcoin~\cite{nakamoto2008bitcoin} and Ethereum~\cite{wood2014ethereum}, respectively. In recent years, payment channel schemes with different features such as re-balancing, throughput maximization, attack resistance, and privacy protection,
have been constructed~\cite{li2020secure, khalil2017revive, papadis2022payment, green2017bolt, malavolta2017concurrency}.

\noindent\textbf{State channel schemes.} 
A state channel enriches the functionality of
a payment channel. Concretely, the users of a state channel can,
besides payments, execute complex smart contracts in an off-chain
way (e.g., voting, auctions) and allow the exchange of states between two or more participants~\cite{dziembowski2018general, close2018forcemove}. The concept of state channel was proposed by Jeff Coleman~\cite{statechannel}. Later, Counterfactual~\cite{coleman2018counterfactual} 
gave a detailed design and Dziembowski {\em et~al.}~\cite{dziembowski2018general} provided formal definitions and security proofs for the general state channel network. ForceMove~\cite{close2018forcemove} is a framework that can support $n$-party participation in a state channel. State channel schemes with
faster payment speeds were developed in~\cite{miller2019sprites, chakravarty2020hydra}.

\noindent\textbf{Virtual channel schemes.} 
Virtual channels enable the creation,
progression, and closing of the channel without interacting with the underlying blockchain. 
%In some sense,
%they thus mimic the functionality offered by payment channels,
%with the difference that they are not created directly over
%the blockchain but instead over two payment channels.
Dziembowski {\em et~al.}~proposed Perun~\cite{dziembowski2019perun}, the first virtual channel scheme  in Ethereum. Later, they presented another scheme in~\cite{dziembowski2019multi}, discussing how to support virtual multi-party state channels. Aumayr {\em et~al.}~\cite{aumayr2021bitcoin} designed a virtual channel compatible with Bitcoin, proving that the establishment of a virtual channel can be independent of smart contracts.

%然而，none of the works mentioned above considers跨链需求， when designing their protocols.
\noindent\textbf{Summary and motivation.} According to the above analysis, one can see that %channels %are developing towards more generality (payment channels to state channels) and scalability (payment/state channels to virtual channels). 
the emergence of state channels broadens the application of payment channels, enabling off-chain channels to provide more services. Virtual channels can effectively reduce the cost of channel network establishment and improve the efficiency of transaction processing.
%can provide more services, reduce on-chain costs and improve transaction processing efficiency.
Even though these channel schemes can successfully enhance the scalability of blockchain systems, they were originally proposed for operations within a single-chain, and cannot be directly extended to support cross-chain operations considering the challenges brought by the problems of UAC and UE. Furthermore, the design of the
current cross-chain solutions that do not rely on third parties, e.g., HTLC,
%their 
targets synchronous networks, while the unbounded latency in asynchronous networks may render them completely fail. %  there are still some unresolved issues (described in Sec.~\ref{sec:introduction}), which makes current channel schemes cannot meet the needs of complex scenarios such as cross-chain. 
Motivated by these considerations,
we propose $\mathsf{Cross}$-$\mathsf{Channel}$ in this paper, which can effectively support efficient and fair atomic cross-chain operations under decentralized asynchronous networks.
	\section{Models and Preliminaries}\label{sec:Preliminaries}
In this section, we first define our system model and threat model, then  %present the assumptions of $\mathsf{Cross\text{-}Channel}$, and finally 
provide preliminaries on  zero-knowledge proof, fair exchange, hashed timelock contracts, and threshold key management

\subsection{Models}\label{model}
In this paper, we consider building a channel between heterogeneous blockchains. Such a channel involves three entities: sender ($\mathcal{S}$), receiver ($\mathcal{R}$), and blockchain miners ($\mathcal{M}$). $\mathcal{S}$ and $\mathcal{R}$ are the two parties of a channel interaction, being responsible for opening and closing the channel,  uploading signature information, etc. %Miners $\mathcal{M}$ are responsible for managing the blockchain with a secure consensus algorithm (e.g. proof of work, proof of stake). 
$\mathcal{M}$ is required to execute a smart contract to determine the legitimacy of the uploaded information, and honest miners would be accordingly rewarded by the blockchain incentive mechanism (just like the main chain of Ethereum). 
\begin{itemize}
	\item Sender $\mathcal{S}$. We assume that they can be arbitrarily malicious, and can act in their best interests. %Furthermore, a sender and a receiver may collude with each other to maximize their benefits.
	\item Receiver $\mathcal{R}$. We assume that they can be arbitrarily malicious, and can act in their best interests.
%	\item Receiver $\mathcal{R}$. We assume that $\mathcal{R}$ can be arbitrarily malicious, and act in their best interests. Furthermore, a receiver may collude with each other to maximize his benefits.
	\item Miner $\mathcal{M}$. Multiple miners follow a secure consensus algorithm to maintain the blockchain. Adversaries cannot compromise the majority of them to bring down the overall blockchain system.
\end{itemize}
All transactions can be divided into two categories, with one being the traditional on-chain transactions (or called transactions), which are confirmed and verified through the blockchain consensus mechanism, and the other being the off-chain transactions (or called receipts), which exist in channels and are verified by nodes within the channel. Some receipts would eventually be packaged into on-chain transactions and update the on-chain states of the nodes in the channel.

$\mathsf{Cross}$-$\mathsf{Channel}$ aims to realize scalability, fairness, and atomicity. % We give detailed explanations and proofs in Sec.~\ref{discussions}. In order to 
To achieve these goals, we next introduce the adopted key technologies.

\subsection{Zero-knowledge Proof: zk-SNARK} \label{zk-snark}
%Zero-knowledge proofs allow one party (the prover) to prove to another (the verifier) that a
%statement is true, without revealing any information beyond the validity of the statement itself~\cite{groth2016size}. 
zk-SNARK (zero knowledge Succinct Non-interactive ARgument of Knowledge) is one type of the zero-knowledge proofs, which allows one party (the prover) to prove to another party (the verifier) that a statement is true, without revealing any information beyond the validity of the statement~\cite{groth2016size}. 
%As a type of zero-knowledge proof, zk-SNARK (zero knowledge Succinct Non-interactive ARgument
%of Knowledge) has found its great value in fair exchange (shown in~\ref{fair-exchange}) and preserving privacy for blockchains~\cite{dziembowski2018fairswap, sasson2014zerocash}. This is mainly due to its
%security features in the succinct proof.
\newtheorem{mydef}{Definition}
\begin{mydef}[zk-SNARK for an $\mathbb{F}$-arithmetic Circuit]
	An $\mathbb{F}$-arithmetic circuit $C$ takes inputs (public inputs $\vec{x}$, private inputs $\vec{w}$ ) from a finite field such as ($\mathbb{F}^n$, $\mathbb{F}^h$), and outputs the result ( $\in \mathbb{F}^l$) based on the circuit logic. A zk-SNARK scheme essentially aims to ensure the satisfaction ($C(\vec{x}, \vec{w}) = 0^l$) of $C$, denoted as $R_C$. The whole process can be represented by a tuple of polynomial-time algorithms $\Pi$ $\overset{\text{def}}{=}$ (Setup, Prove, Verify):
	\begin{itemize}
		\item  ${\mathtt{Setup}}(1^{\lambda}, C) \rightarrow (\mathsf{pk, vk})$.
		The algorithm $\mathtt{Setup}$ takes a security parameter $1^{\lambda}$ and a circuit $C$ as inputs to obtain the key pair $(\mathsf{pk, vk})$, where $\mathsf{pk}$ is the proving key for proof generation and $\mathsf{vk}$ is the verification key for proof verification.  The pair $(\mathsf{pk, vk})$ constitutes the common reference string $\mathsf{crs}$.
		\item $\mathtt{Prove}(\mathsf{pk}, \vec{x}, \vec{w}) \rightarrow \pi $.
		%On input $\vec{x}, \vec{w}$, generate a proof $\pi$ with the proving key $\mathsf{pk}$.
		The algorithm takes as inputs the proving key $\mathsf{pk}$, the public inputs $\vec{x}$ and the private inputs $\vec{w}$ to generate a succinct zero-knowledge proof $\pi$.
		\item $ \mathtt{Verify}(\mathsf{vk}, \vec{x}, \pi)  \rightarrow \mathsf{1/0}$.
		%It verifies the proof $\pi$ with verification key $vk$ and the input $\vec{x}$.The algorithm returns 1 if the verification is successful and 0 otherwise.
		The algorithm verifies $\pi$ based on the verification key $\mathsf{vk}$ and public inputs $\vec{x}$. It returns 1 if the verification is successful and 0 otherwise.
	\end{itemize}
\end{mydef}
%In the algorithm $\Pi.\mathsf{Setup}$, ${\mathsf{crs}}$ is public, which means that anyone with circuit $C$ can verify a proof.
Given a security parameter $\lambda$ and a circuit $C$ with a relation $R_C$, an honest $\mathcal{S}$ can generate a proof $\pi$ to convince $\mathcal{R}$ for every pair $(\vec{x}, \vec{w})$ $\in$ $R_C$. In the algorithm $\Pi.\mathsf{Setup}$, ${\mathsf{vk}}$ and $C$ are public, which means that anyone with $\vec{x}$ can verify a proof.
          
%\subsection{Fair Exchange: Digital Commodity Payment Based on zk-SNARK}
\subsection{Fair Exchange Based on zk-SNARK}\label{fair-exchange}

Fair exchange refers to the scenario where users exchange currency for digital commodities, e.g., digital assets and valuable information.
A fair exchange protocol was designed to guarantee that the exchange is executed in a fair way~\cite{dziembowski2018fairswap}. 
zk-SNARK (see Sec.~\ref{zk-snark}) is one of the key technologies to realize fair exchange. It can help $\mathcal{S}$ protect the privacy of information content while proving its authenticity. Next, we give the definition of the circuit used for realizing fair exchange.
\newtheorem{mydef2}{Definition}
\begin{mydef}[The Circuit for Fair Exchange]\label{fair-exchange-circuit}
	 The whole process can be represented by a tuple of polynomial-time algorithms $\Upsilon$ $\overset{\text{def}}{=}$ ($\mathtt{DataAuth}$, $\mathtt{KeyAuth}$):
	\begin{itemize}
		\item $\mathtt{DataAuth}(m) \rightarrow h(m) $.
		 The algorithm takes the digital commodity (or plaintext) $\mathsf{m}$ as input, and computes the authenticator $h(m)$, which is the hash result of the digital commodity.
		\item $ \mathtt{KeyAuth}(k, m) \rightarrow (\overline{m}, h(k))$.
		The algorithm takes the encryption key $\mathsf{k}$ and the digital commodity $\mathsf{m}$ as inputs, and generates the encrypted digital commodity (or ciphertext) $\overline{m}$ as well as the hash result of the encryption key $h(k)$.
	\end{itemize}
\end{mydef}
This circuit is illustrated in Fig.~\ref{Fig:fair-exchange-circuit}. $\mathcal{S}$ can use the circuit to generate a zero-knowledge proof based on the algorithm $\Pi.\mathtt{Prove}$, which proves the authenticity of the encrypted information and the encryption key.
\begin{figure}[htb]
	\centering
	\includegraphics[width=0.6\textwidth]{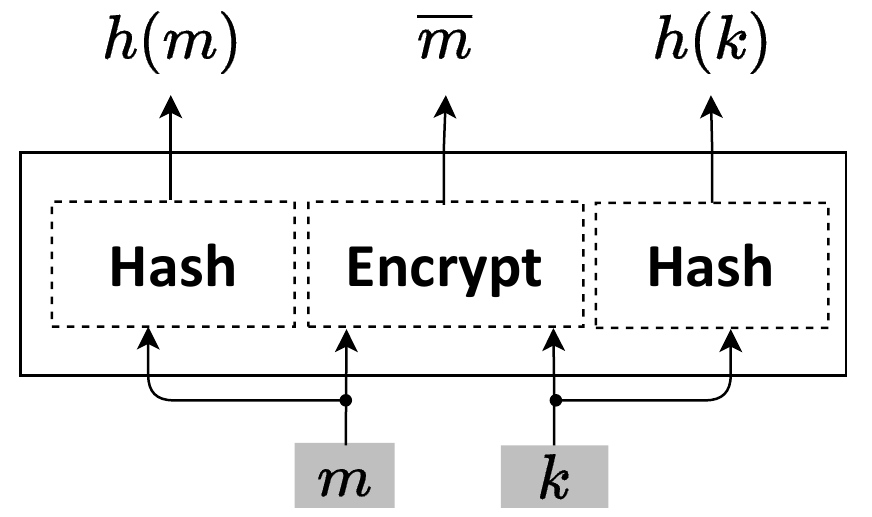}
	\caption{The logic diagram of the fair exchange circuit. The parameters with gray background are private ones protected with zk-SNARK.}
	\label{Fig:fair-exchange-circuit}
		\vspace{-0.2cm} 
\end{figure}

\subsection{Hashed Timelock Contracts}\label{htlc}
The Hashed Timelock Contract (HTLC) was first applied and implemented in Bitcoin's Lightning network~\cite{poon2016bitcoin}, which aims to ensure the atomicity of cross-chain asset exchanges.
HTLC requires that both sides of an interaction (e.g., $\mathcal{S}$ and $\mathcal{R}$) have accounts in each blockchain (i.e., accounts $\mathcal{S^\alpha}$ and $\mathcal{R^\alpha}$ in chain $\alpha$, and accounts $\mathcal{S^\beta}$ and $\mathcal{R^\beta}$ in chain $\beta$, for both $\mathcal{S}$ and $\mathcal{R}$). 
The smart contracts\footnote{
Some blockchain systems such as Bitcoin do not support smart contracts~\cite{zou2019smart}. In such a case, HTLC is implemented with other mechanisms such as scripting~\cite{script}.
For convenience, we use smart contracts to represent true smart contracts as well as other techniques such as scripting when presenting HTLC in this study.} in blockchain $\alpha$ and $\beta$ are denoted by $\xi^\alpha$ and $\xi^\beta$, respectively.

\newtheorem{mydef4}{(Definition)}
\begin{mydef}[The Process of HTLC]
    Assume that account $\mathcal{S^\alpha}$ and account $\mathcal{R^\beta}$ intend to exchange assets with each other. The whole interaction process in HTLC is divided into three steps: $\mathsf{Lock}$, $\mathsf{Update}$, and $\mathsf{Refund}$.
	\begin{itemize}
		\item $\mathsf{Lock:}$ First, $\mathcal{S^\alpha}$ selects a random 256-bit integer as the preimage pre and computes its hash value h(pre). Then in the smart contract $\xi^\alpha$, $\mathcal{S^\alpha}$ opens h(pre), employs h(pre) to lock its asset sent to $\mathcal{R^\alpha}$, and sets a timer $T_3$. Similarly, in the smart contract $\xi^\beta$, $\mathcal{R^\beta}$ locks its asset sent to $\mathcal{S^\beta}$ with the same h(pre) and sets a timer $T_4$, where $T_3>T_4$. 
		
	    \item $\mathsf{Update:}$ $\mathcal{S^\beta}$ offers pre to $\xi^\beta$ within $T_4$ to unlock the asset sent by $\mathcal{R^\beta}$. 
	    After $\mathcal{R^\beta}$ learns pre, $\mathcal{R^\alpha}$ provides pre to $\xi^\alpha$ within $T_3$ to unlock the asset sent by $\mathcal{S^\alpha}$.
	    
	    \item $\mathsf{Refund:}$ If the time exceeds $T_4$ and $S^\beta$ does not provide pre, the locked asset in $\xi^\beta$ would be returned to $R^\beta$. In this case, since $\mathcal{R^\alpha}$ does not know pre (only $\mathcal{S}$ has pre), $\mathcal{R^\alpha}$ cannot provide pre within the specified time $T_3$ in  blockchain $\alpha$. When $T_3$ times out, the locked asset would be returned to $S^\alpha$.
	\end{itemize}
\end{mydef}

Note that the information in accounts $\mathcal{R^\alpha}$ ($\mathcal{S^\alpha}$) and $\mathcal{R^\beta}$ ($\mathcal{S^\beta}$) are shared because they both belong to the same entity $\mathcal{R}$ ($\mathcal{S}$). Therefore, when $\mathcal{R^\beta}$ learns $pre$ in blockchain $\beta$, $\mathcal{R^\alpha}$ can send $pre$ to the smart contract in blockchain $\alpha$.
Also note that $T_3>T_4$ is necessary in order to ensure atomicity. %, the time threshold $T_3$ in $\alpha$ needs to be longer than the time threshold $T_4$ in $\beta$  (related discussion is shown in Sec.~\ref{discussions}).
Nevertheless, in asynchronous networks, we find that HTLC may not guarantee atomicity due to network delay. Therefore, we introduce an incentive mechanism to overcome this problem in Sec.~\ref{hierarchical-channel-scheme}. % and Sec.~\ref{discussions} for more details).

\subsection{Pedersen’s Verifiable Secret Sharing}\label{vss}
Pedersen’s verifiable secret sharing scheme does not need any trusted third party, which enables $n$ participants to share a secret in a completely decentralized way~\cite{pedersen1991non}.

\newtheorem{mydef3}{(Definition)}
\begin{mydef}[Pedersen’s $(t,n)$-VSS]
	Let $\mathbb{G}_q$ be a $q$-order subgroup of the prime $P$, with $g$ and $h$ being generators of $\mathbb{G}_q$. Let $s$ be the shared secret, $\mathcal{O}$ the owner of $s$, $n$ the number of participants, $t$ the threshold value, and $\mathsf{U_i}$  the $i$-th participant. Define Pedersen commitment as $E(a,b)=g^ah^b$. Then the whole process can be divided into three steps: $\mathsf{Share}$, $\mathsf{Verify}$ and $\mathsf{Recover}$.
	\begin{itemize}
		\item $\mathsf{Share}$: First, $\mathcal{O}$ selects a random number $r$, computes commitment $E(s,r)=g^sh^r$, and opens $E(s,r)$. Then, $\mathcal{O}$ selects $t-1$ random numbers $a_i, i\in[1,t-1]$, constructs a polynomial $f(x)=s + \sum_{i=1}^{t-1}a_ix^i$, and  computes $s_i=f(i)$. Next, $\mathcal{O}$ selects another set of random numbers $b_i, i\in[1,t-1]$, calculates $E_{a_i}=g^{a_i}h^{b_i}$, and opens them. Finally, $\mathcal{O}$ constructs a polynomial $g(x)=r + \sum_{i=1}^{t-1}b_ix^i$, computes $r_i=g(i)$, and sends the $i$th secret share ($s_i,r_i$) to $\mathsf{U_i}$.
		\item $\mathsf{Verify}$: When $\mathsf{U_i}$ receives ($s_i,r_i$), it computes $E(s_i,r_i)$ and $\prod_{j=0}^{t-1}E_j^{i^j}$, where  $E_j^{i^j}=g^{a_ji^i}h^{b_ji^i}$. If the computed $E(s_i,r_i)$ and $\prod_{j=0}^{t-1}E_j^{i^j}$ are equal, the received $s_i$ is correct.
		\item $\mathsf{Recover}$: When at least $t$ participants share the secret correctly and contribute their shares, the secret can be recovered by Lagrange polynomial interpolation, i.e. $s=\sum_{i=1}^ts_i\prod_{1 \leq j \leq t, j \neq i} \frac{i}{i-j} $.
	\end{itemize}
\end{mydef}
%In all steps, one can find that $\mathsf{(t,n)\text{-}VSS}$ scheme does not need any trusted third party. 
The participants can verify the validity of the received shares in step $\mathsf{Verify}$, so as to detect the invalid messages sent by adversaries.

%Note that when we adopt $\mathsf{(t,n)\text{-}VSS}$ in  $\mathsf{Cross}$-$\mathsf{Channel}$, we ensure that $t/n$ strictly follow the fault-tolerance property of the underlying blockchain consensus algorithm. For example, if PoW (the Proof-of-Work consensus) is employed, we require that $t/n>1/2$. This requirement is necessary to guarantee that the secret can be correctly recovered, as long as the underlying blockchain system remains secure.

	\newtheorem{assumption}{Assumption}
\newtheorem{lemma}{Lemma}
\newtheorem{theorem}{Theorem}
\section{The Cross-Channel}\label{sec:Cross-Channel}
In this section, we first provide an overview on $\mathsf{Cross}$-$\mathsf{Channel}$, then present the protocol in detail, and finally analyze its scalability, fairness and atomicity.

%TODO:Overview 详细+改图   --Yihao 0929
%完成  --Yihao 0930
\subsection{Overview} \label{overview}

\begin{figure*}[htb]
	\centering
	\includegraphics[width=0.85\textwidth]{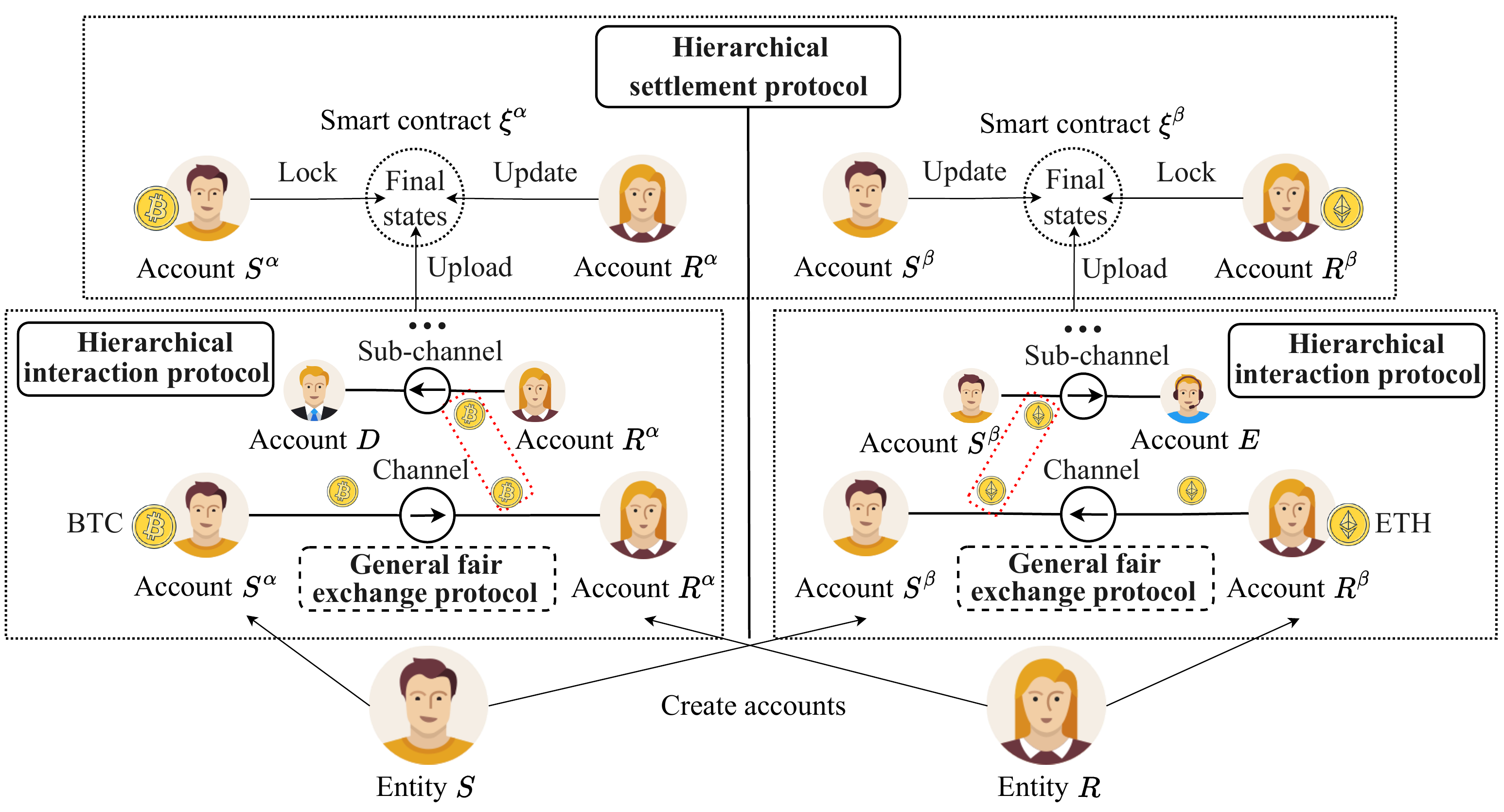}
	\caption{The cross-chain currency exchange procedure with $\mathsf{Cross}$-$\mathsf{Channel}$. The whole process includes three protocols, i.e, the hierarchical interaction protocol, the general fair exchange protocol, and the hierarchical settlement protocol. The hierarchical interaction protocol contains a hierarchical channel structure that solves the UAC problem. The general fair exchange protocol solves the UE problem that occurs during the interaction in the (sub-)channel, especially for the encrypted information exchange scenario. The hierarchical settlement protocol adopts an improved HTLC protocol, overcoming the impact of high latency in asynchronous networks while ensuring the correctness of cross-chain settlements.}
	\label{Fig:overview}
\end{figure*}
%举例说明整个协议的大体过程
$\mathsf{Cross}$-$\mathsf{Channel}$ is an efficient channel scheme that supports complex services such as cross-chain. %, including a hierarchical architecture, an off-chain general fair exchange protocol $\Theta$ and a cross-chain settlement smart contract $\xi$.
For the sake of convenience, we use an example (currency exchange) to illustrate the general process of $\mathsf{Cross}$-$\mathsf{Channel}$ (shown in Fig.~\ref{Fig:overview}). 
In this example, an entity $\mathcal{S}$ has an account $\mathcal{S^\alpha}$ in Bitcoin $\alpha$, and an entity $\mathcal{R}$ has an account $\mathcal{R^\beta}$ in Ethereum $\beta$.
$\mathcal{S}$ and $\mathcal{R}$
attempts to frequently exchange $\mathcal{S^\alpha}$'s Bitcoins (BTC) for $\mathcal{R^\beta}$'s Ether (ETH).
In order to achieve the above goal, $\mathcal{S}$ needs to creates an account $\mathcal{S}^\beta$ in $\beta$ to get $\mathcal{R^\beta}$'s ETH, and $\mathcal{R}$ also needs to create an account $\mathcal{R}^\alpha$ in $\alpha$ to get $\mathcal{S^\alpha}$'s BTC. The whole process can be summarized as follows.
%$\mathcal{S}'s$ accounts in $\alpha$ and $\beta$ are remarked as $\mathcal{S}^\alpha$ and $\mathcal{S}^\beta$, where the account $\mathcal{S}^\alpha$ has enough BTC. $\mathcal{R}'s$ accounts in $\alpha$ and $\beta$ are denoted as $\mathcal{R}^\alpha$ and $\mathcal{R}^\beta$, where the account $\mathcal{R}^\beta$ has enough ETC.

First, $\mathcal{S}$ and $\mathcal{R}$ need to establish channels in $\alpha$ and $\beta$. Specifically, two accounts in the same blockchain ($\mathcal{S}^\alpha$ with $\mathcal{R}^\alpha$, or $\mathcal{S}^\beta$ with $\mathcal{R}^\beta$) execute the hierarchical interaction protocol $\Psi$ to establish a channel, then send currency in this channel. The channel can be a hierarchical one with multiple sub-channels in order to spend the unsettled amounts (e.g., accounts $D$ and $E$ in Fig.~\ref{Fig:overview}) based on $\Psi$. Currency exchanges within the channel follow the general fair exchange protocol $\Theta$, which also supports fair exchange (i.e., exchange currency with encrypted information) and encrypted information exchange.
Finally, when the channel needs to be closed, all involved accounts in the channel, including those for the sub-channels,  upload their final states to the corresponding smart contracts, and the hierarchical settlement protocol $\Phi$ is executed to complete the settlement. Note that, we adopt an improved HTLC in $\Phi$ to ensure the correctness and atomicity of the cross-chain settlement. 

In the following two subsections, we detail the hierarchical interaction protocol with settlement (Sec.~\ref{hierarchical-channel-scheme}) and the general fair exchange protocol (Sec.~\ref{general-fair-exchange}).

\subsection{Hierarchical Channel Design} \label{hierarchical-channel-scheme}
\begin{figure}[htb]
	\centering
	\includegraphics[width=0.8\textwidth]{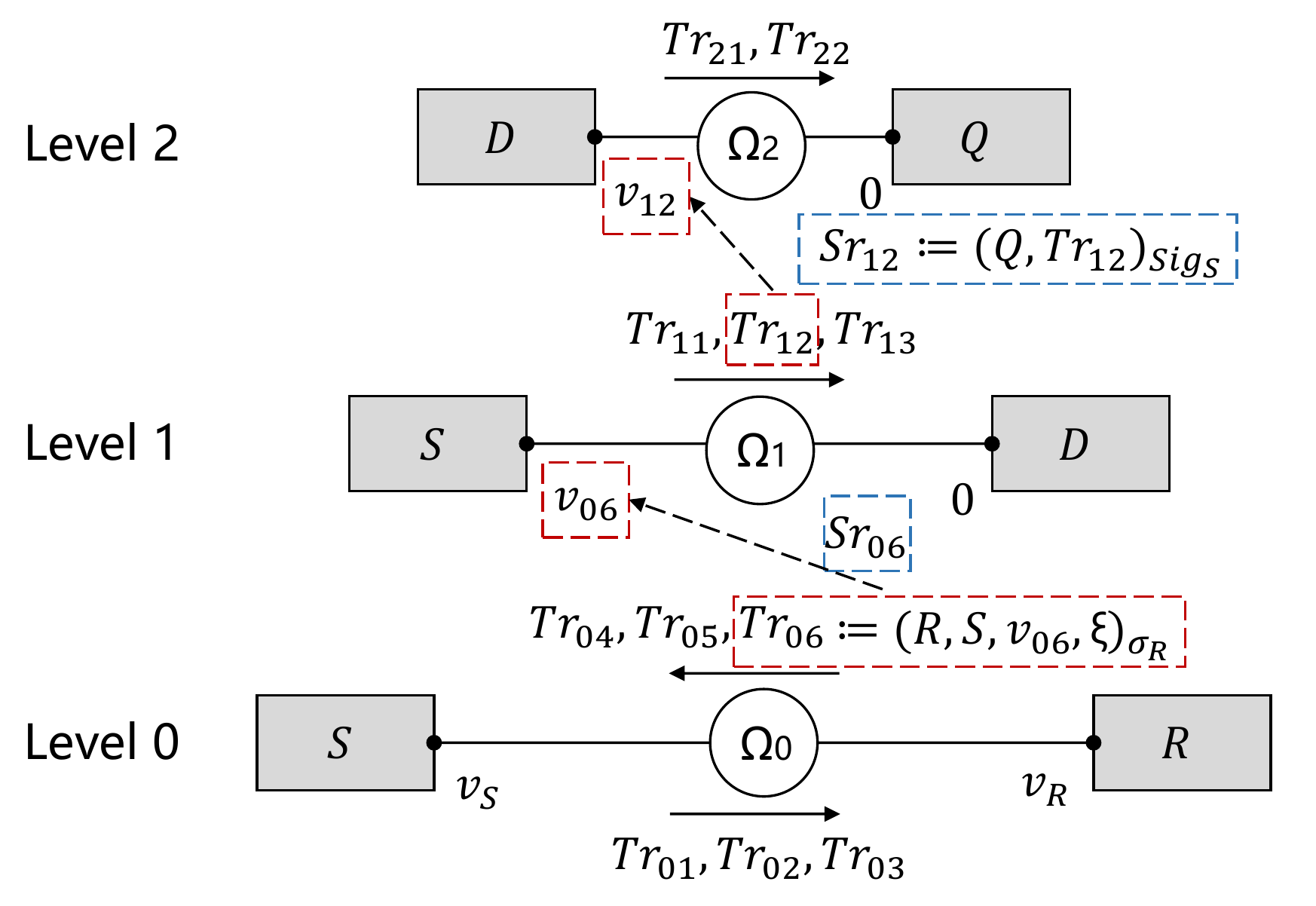}
	\caption{An illustration example of a 3-level hierarchical channel.}
	\label{Fig:hibe-channel}
\end{figure}
We present a hierarchical interaction protocol $\Psi$ with a settlement protocol $\Phi$ to solve the Unsettled Amount Congestion (UAC) problem. 
The settlement protocol $\Phi$ can be written into smart contracts and verified by miners $\mathcal{M}$. 

\noindent{\bf Hierarchical Interaction Protocol $\Psi$.} Fig.~\ref{Fig:hibe-channel} illustrates an example  hierarchical channel. One can see that the whole structure has three levels, which are marked as $\mathsf{Level}$ 0--2. 
At $\mathsf{Level}$ 0, $\mathcal{S}$ and $\mathcal{R}$ send $\mathsf{Tx}_{\text{Open}}$ to the smart contract $\xi$ to establish a channel $\Omega_0$ with initial state ($\mathsf{v_\mathcal{S}, v_\mathcal{R}}$), which can be denoted as $[\mathcal{S} \mapsto \mathsf{v_S}$, $\mathcal{R} \mapsto \mathsf{v_R}]_{\Omega_0}$, meaning that $\mathcal{S}$ has $\mathsf{v_\mathcal{S}}$ in $\Omega_0$, $\mathcal{R}$ has $\mathsf{v_\mathcal{R}}$ in $\Omega_0$, and the state of $\Omega_0$ is ($\mathsf{v_\mathcal{S}, v_\mathcal{R}}$).
\begin{equation}\notag
	\mathsf{Tx}_{\text{Open}} \overset{\text{def}}{=} (\mathsf{From: \mathcal{S / R}; To: \xi; v_\mathcal{S}/v_\mathcal{R}}).
\end{equation}
For the sake of convenience, $\mathsf{v_S}$ and $\mathsf{v_R}$ can be the coins deposited by $\mathcal{S}$ and $\mathcal{R}$ (for some $\mathsf{v_S}$, $\mathsf{v_R}$ $\in$ $\mathbb{R}$ $\geq$ 0). 
%Yihao: 签名符号：此部分共涉及三个签名，第一个签名是通道内的签名例如在$\Omega_0$中，一方给另一方发送Tr，Tr内包含的签名我们用\sigma表示
%第二、三个签名是建立子通道的签名，例如建立$\Omega_1$需要$\Omega_0$中交互双方的签名，我们用Sig表示。
In $\Omega_0$, $\mathcal{S}$ and $\mathcal{R}$ can send transaction receipts \{$Tr$\} to each other, such as $Tr_{01\text{--}06}$ in Fig.~\ref{Fig:hibe-channel}. Let $Tr \overset{\text{def}}{=} (\mathsf{snd},\mathsf{rcv}, \mathsf{v}, \xi)_{\mathsf{\sigma_{snd}}}$, meaning that $\mathsf{snd}$ transfers amount $\mathsf{v}$ to $\mathsf{rcv}$ via smart contract $\xi$, where $\mathsf{\sigma_{snd}}$ is the message signature. %where $\mathsf{\sigma_{snd}}$ is the sender's signature for the transfer amount $\mathsf{v}$, $\mathsf{rcv}$ represents the receiver address, and $\xi$ is the smart contract address. 
Note that $\mathsf{snd}$ and $\mathsf{rcv}$ can be omitted if clear from context. 

To spend an unsettled amount in $\Omega_0$, the two parties of $\Omega_0$ can negotiate to open a sub-channel. For example, $\mathcal{S}$ can send a request to $\mathcal{R}$ asking for the permission to open a sub-channel $\Omega_1$ with $\mathcal{D}$ to spend the unsettled amount in $Tr_{06}$.   %This request message has a signature $\mathsf{Sig_{\mathcal{S}}}$.  
If $\mathcal{R}$ permits, 
%First, $\mathcal{S}$ selects a transaction receipt (e.g., $Tr_{06}$) and send a request to $\mathcal{R}$. 
%
it would generate a sub-channel receipt $Sr_{06}$ and send it to $\mathcal{S}$, where $Sr_{06} \overset{\text{def}}{=} (\mathcal{D}, Tr_{06})_{\mathsf{Sig_{\mathcal{R}}}}$, with  $\mathsf{Sig_{\mathcal{R}}}$ being $\mathcal{R}$'s signature for $Sr_{06}$, and $\mathcal{D}$ the address of the participant with which $\mathcal{S}$ would construct a sub-channel to spend the unsettled amount in $Tr_{06}$. 
Then, $\mathcal{S}$ sends $Sr_{06}$ to $\mathcal{D}$, who needs to verify the legitimacy of $Sr_{06}$ based on $\mathsf{Sig_{\mathcal{R}}}$ and ensure the correctness of $\mathsf{v_{06}}$ according to the $\mathsf{\sigma_{\mathcal{R}}}$ carried by $Tr_{06}$.
When the verification is successful, $\mathcal{S}$ takes $\mathsf{v_{06}}$ as its initial balance to open sub-channel $\Omega_1$ with $\mathcal{D}$, i.e., $[\mathcal{S} \mapsto \mathsf{v_{06}}, \mathcal{D} \mapsto 0]_{\Omega_1}$. Note that $\Omega_1$ is a sub-channel that is constructed particularly for the spending of the unsettled amount in $Tr_{06}$ -- no other transactions between $\mathcal{S}$ and $\mathcal{D}$ are allowed. Following the same procedure, the parties ($\mathcal{S}$ and $\mathcal{R}$) in $\mathsf{Level}$ 0 can choose another $Tr$ to create another new sub-channel and the parties ($\mathcal{S}$ and $\mathcal{D}$) in $\mathsf{Level}$ 1 can also generate sub-channels. 
For example, as shown in Fig.~\ref{Fig:hibe-channel}, $\mathcal{D}$ establishes a sub-channel $\Omega_2$ with user $\mathcal{Q}$ based on $Tr_{12}$. It is worthy of noting that all operations related to a sub-channel are off-chain, which means that smart contract is not involved thus conserving blockchain resources. 

\noindent{\bf Hierarchical settlement protocol $\Phi$.} We design a new protocol $\Phi$ %(shown in Algorithm~\ref{A:HCS})
and implement it in smart contract $\xi$ to support settlement. 
Not that, $\Phi$ can be adopted for both intra-chain and cross-chain channel settlement, which differ slightly.

We use the same example shown in Fig.~\ref{Fig:hibe-channel} to demonstrate the procedure for intra-chain settlement. 
First, $\mathcal{S}$ and $\mathcal{R}$ send requests to $\xi$ to close the hierarchical channel. After receiving the channel closing requests, $\xi$ sets a timer $T_2$ %to collect the final states in the hierarchical channel 
and broadcasts this closing event to all blockchain participants and miners.
This message also requires the users ($\mathcal{S}$, $\mathcal{R}$, $\mathcal{D}$ and $\mathcal{Q}$) to compute their final states \{$f$\} based on the corresponding related receipts ($Tr$ and $Sr$). 
For instance, $\mathcal{S}$ and $\mathcal{R}$ need to compute their final states based on $Tr_{01\text{--}06}$, $Sr_{06}$, $Tr_{11\text{--}13}$, and $Sr_{12}$.
Then, each participant packages its $f$ and the related receipts into a $\mathsf{Tx}_{\text{Close}}$ message and sends $\mathsf{Tx}_{\text{Close}}$ to $\xi$ within $T_2$. 
%Considering the high latency of an asynchronous network, %and the incentive manipulation attacks
%some participant in the channel may fail to upload its $\mathsf{Tx}_{\text{Close}}$ within $T_2$, which makes the settlement result inaccurate. %For example,
%To overcome this problem, we set another timer $T_3$ in $\xi$, allowing
%the participant who fails in $T_2$ to request the miners who have received $\mathsf{Tx}_{\text{Close}}$ in $T_2$ to upload to other miners during $T_3$ on its behalf. 
\begin{equation}\notag
	\mathsf{Tx}_{\text{Close}} \overset{\text{def}}{=} (\mathsf{From: Snd; To: \xi;}\; f, \{Sr\}).
\end{equation}
%{\bf Yihao, we need to figure out whether we need $t_3$. if $t_3$ is not needed, change the $t_3$ to $t_2$ in the following paragraph.}

When $T_2$ times out, $\xi$ verifies the uploaded data from $\mathsf{Level}$ 0, i.e., the correctness of the signatures and account balances. Note that $\mathsf{Tx}_{\text{Close}}$ contains all sub-channel receipts agreed by $\mathsf{Snd}$, which are used by $\xi$ to check the correctness of account balances. If the verification is successful, $\xi$ would verify the data from the next sub-level, i.e., $\mathsf{Level}$ 1. If the verification succeeds, $\xi$ continues to verify the next higher level. If $\xi$ fails at any level, all sub-channels in that level and above would automatically fail. Such a failure drives $\xi$ to adjust the final state of each participant in all successful levels based on the received $\{Sr\}$'s. Finally, the miners update the corresponding on-chain states according to the settlement results.

To support cross-chain settlement, the hierarchical settlement protocol $\Phi$ adopts the HTLC protocol (introduced in Sec.~\ref{htlc}) to ensure the atomicity of the interaction. 
As shown in Fig.~\ref{Fig:overview}, $\mathcal{S}$ and $\mathcal{R}$ need to have accounts in both blockchain $\alpha$ and $\beta$ for cross-chain operations. Based on the hierarchical interaction protocol and the intra-chain hierarchical settlement protocol mentioned above, accounts in the same chain ($\mathcal{S^\alpha}$ and $\mathcal{R^\alpha}$, $\mathcal{S^\beta}$ and $\mathcal{R^\beta}$) establish channels, send $Tr$, create $Sr$ to establish sub-channels, and upload the final state of each (sub-)channel when the channel needs to be closed. 
Unlike the intra-chain settlement, which updates the on-chain states immediately, for cross-chain, when $\xi$ determines that the final states are valid, it runs the HTLC protocol to make $\mathcal{S^\alpha}$ and $\mathcal{R^\beta}$ lock their final states based on the step HTLC.$\mathsf{Lock}$, then $\mathcal{S}$ and $\mathcal{R}$ complete the final settlement according to the step HTLC.$\mathsf{Update}$ or HTLC.$\mathsf{Refund}$.

%摘自Sec.D对于异步的讨论     --Yihao 1001
Particularly, in the step HTLC.$\mathsf{Refund}$ shown in Sec.~\ref{htlc}, if no one submits $\mathsf{pre}$ before $T_3$ or $T_4$ times out, the smart contract would not update the states. However, in an asynchronous network, affected by the high latency, after $\mathcal{S^\beta}$ provides $\mathsf{pre}$ to update the states in blockchain $\beta$, $\mathcal{R^\alpha}$ may fail to upload $\mathsf{pre}$ within $T_3$, which breaks the atomicity of HTLC. %and makes HTLC unsuitable for asynchronous networks. 
Therefore, in order to make HTLC suitable for asynchronous networks, we set a timer $T_5$ in $\xi^\alpha$, during which any miner can help $\mathcal{R^\alpha}$ provide $\mathsf{pre}$ to get rewards from $\mathcal{R^\alpha}$.

Note that in more complex scenarios such as the encrypted information exchange, in addition to realizing the settlement mentioned above, the hierarchical settlement protocol needs to further accomplish the fair exchange of keys, which are detailed in Sec.~\ref{cross-channel and applications}.
%[$\mathcal{S} \mapsto \mathsf{v}_\mathcal{S}-\mathsf{v}_{01\text{--}03}+\mathsf{v}_{04\text{--}06}-\mathsf{v}_{11\text{--}13}$, $\mathcal{R} \mapsto \mathsf{v}_\mathcal{R}+\mathsf{v}_{01\text{--}03}, \mathcal{D} \mapsto 0+\mathsf{v}_{11\text{--}13}-\mathsf{v}_{21\text{--}22}$, $\mathcal{Q} \mapsto 0+\mathsf{v}_{21\text{--}22}$]. 

\subsection{General Fair Exchange Protocol} \label{general-fair-exchange}
\begin{figure}[htb]
	\centering
	\includegraphics[width=\textwidth]{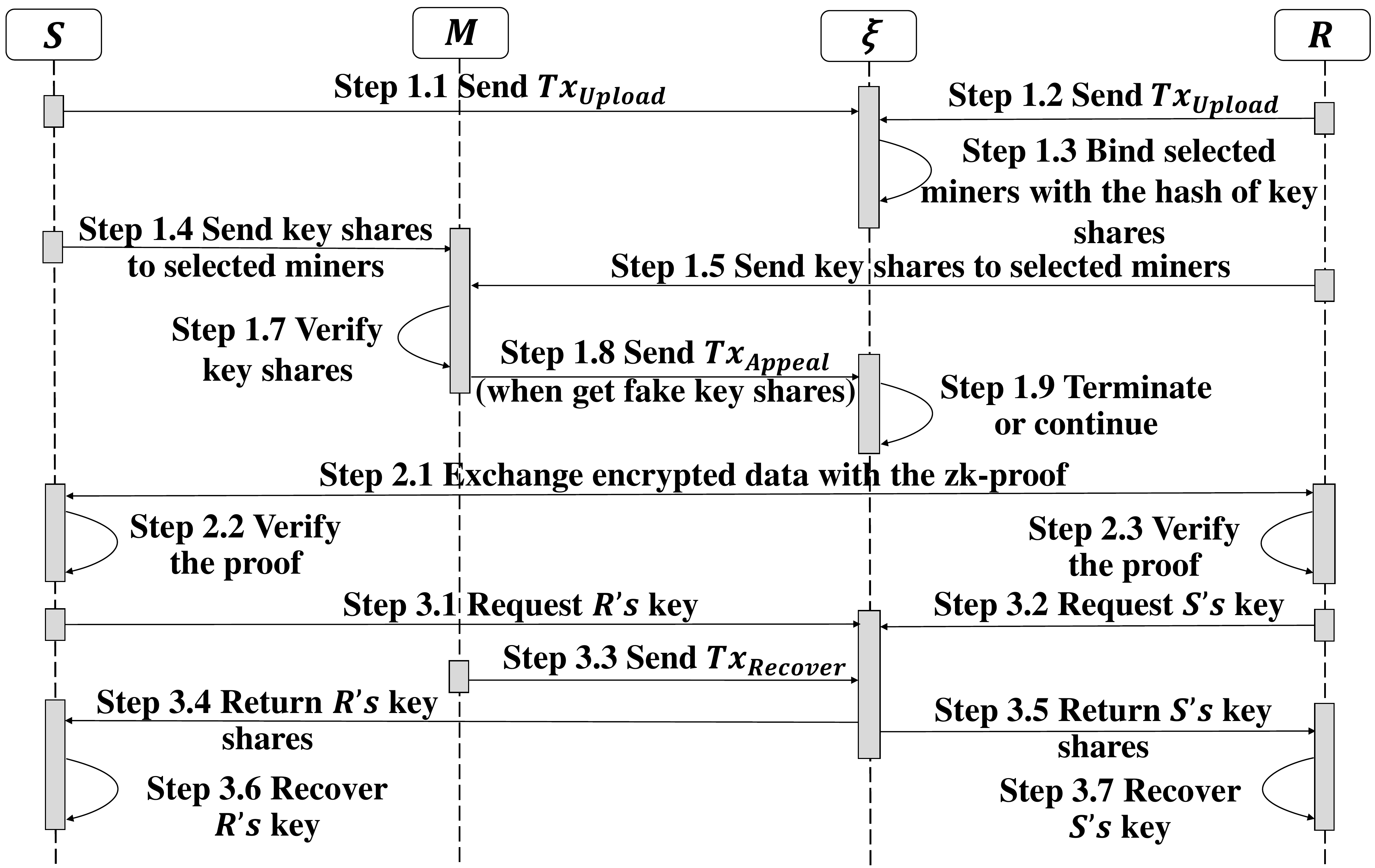}
	\caption{ The sequence diagram of the general fair exchange protocol. Steps 1.1--1.9 demonstrate the $\Theta.$Share process, while steps 2.1--2.3 illustrating the $\Theta.$Exchange process and steps 3.1--3.7 standing for the $\Theta.$Recover process.
	}
	\label{Fig:fair-exchange}
\end{figure}
%每开启一次通道需要换一次密钥
In this subsection, we propose a general fair exchange protocol $\Theta$ to solve the Unfair Exchange (UE) problem, which can guarantee the fairness of encrypted information exchange (EIE). 
%We label it with $\Theta$ and 
The whole protocol %is shown in Fig.\ref{Fig:fair-exchange} and 
involves four steps: {\bf Setup}, {\bf Share}, {\bf Exchange}, and {\bf Recover}, which are demonstrated in Fig.~\ref{Fig:fair-exchange}.

\noindent{\bf $\Theta.$Setup.} First, sender $\mathcal{S}$ builds a circuit $C_\Theta$ based on Fig.~\ref{Fig:fair-exchange-circuit2}.
Compared with the traditional fair exchange circuit (shown in Fig.~\ref{fair-exchange-circuit}), $C_\Theta$ appends $n$ key shares as private inputs and adds a key recover function (details shown in Sec.~\ref{vss}.$\mathsf{Recover}$) to recover the encryption key $k$. 
The reason for designing this circuit lies in that, in $\Theta$, we adopt the $(t,n)$-VSS protocol (shown in Sec.\ref{vss}) to divide the encryption key into $n$ key shares. However, the input of traditional circuit (shown in Fig.~\ref{fair-exchange-circuit}) is the key itself, which cannot prove the correctness of the key shares.
Thus we propose a new circuit $C_\Theta$, which can prove the correctness of not only the key but also the key shares without exposing any key-related information. 
%Without exposing the key-related information, the above circuit design for zk-SNARK can not only prove the correctness of the key but also the correctness of the key share, so as to better support the protocol $\Theta$ ().
Besides $C_\Theta$, $\mathcal{S}$ needs to generate a security parameter $1^{\lambda}$, and takes $C_\Theta$ and $1^{\lambda}$ as inputs of $\Pi.\mathsf{Setup}$ to construct the common reference string ($\mathsf{pk, vk}$).
%\begin{equation}\notag
%	\mathsf{crs} \leftarrow \Pi.\mathsf{Setup}(1^{\lambda}, C_\mathcal{S})
%\end{equation}
\begin{figure}[htb]
	\centering
	\includegraphics[width=\textwidth]{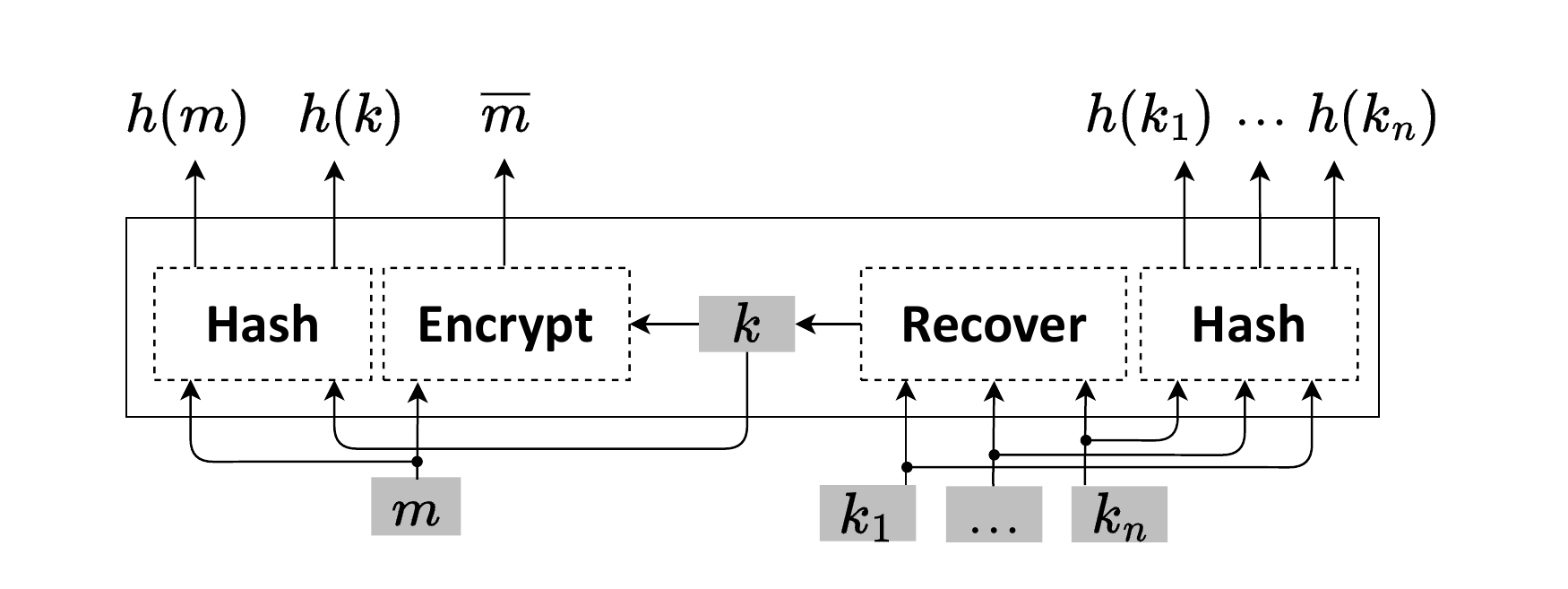}
	\caption{The logic diagram of the circuit used for the general fair exchange protocol $\Theta$. The parameters with gray background are private ones protected with zk-SNARK.}
	\label{Fig:fair-exchange-circuit2}
		\vspace{-0.2cm} 
\end{figure}

%TODO：问题较大，重写该部分协议    --Yihao 20220913
% 完成 --Yihao 0915
\noindent{\bf $\Theta.$Share.} This process is marked Step 1.1 to Step 1.9 in Fig.~\ref{Fig:fair-exchange}. First, $\mathcal{S}$ generates an encryption key $\mathsf{k_\mathcal{S}}$ and divides it into $\mathsf{n}$ shares based on $\mathsf{(t,n)\text{-}VSS.Share}$.
% and distributes them with $\mathcal{S}$'s signature to $\mathsf{n}$ miners who are randomly selected by smart contracts.
%
Then, $\mathcal{S}$ hashes $\mathsf{k_\mathcal{S}}$ and each individual key share $\mathsf{k_{\mathcal{S}i}(i \in [1,n])}$, packages them as well as the number of key shares $n$ and the key threshold $t$, into transaction $\mathsf{Tx}_{\text{Upload}}$. 
Next, $\mathcal{S}$ sends $\mathsf{Tx}_{\text{Upload}}$ to the smart contract $\xi$. 
Following the same procedure the receiver $\mathcal{R}$ packages $\mathsf{h(\mathsf{k_\mathcal{R}}})$, $n$, $t$, and  $\mathsf{h({k_{\mathcal{R}{[1:n]}}})}$ into its $\mathsf{Tx}_{\text{Upload}}$ and sends it to $\xi$. 
\begin{equation}\notag
	\mathsf{Tx}_{\text{Upload}} \overset{\text{def}}{=} (\mathsf{From: \mathcal{S / R}; To: \xi; \; h(\mathsf{k_\mathcal{S / R}}), n ,t, \mathsf{h(k_{\mathcal{S / R}[1:n]}})}).
\end{equation}

When $\xi$ receives $\mathsf{Tx}_{\text{Upload}}$ from both parties, it would randomly select $\mathsf{n}$ miner addresses with a serial number $\mathsf{sn}$, bind each miner address with the hashes of two unique key shares (one from $\mathcal{S}$ and one from $\mathcal{R}$), and open these bindings. %and its recipient address. 
Then, $\mathcal{S}$ and $\mathcal{R}$ sign their key shares with $\mathsf{sn}$, and distribute them to the selected miners based on the bindings.
After receiving the key shares, each miner verifies their legitimacy based on $\mathsf{(t,n)\text{-}VSS.Verify}$. Besides that, each miner recomputes the hashes of the received key shares and compares them with those in $\xi$ to detect possible errors.
%the received key share $\mathsf{h'({k_{i}}) (i \in [1,n])}$ and compares $\mathsf{h'(k_{i})}$ with $\mathsf{h(k_{i})}$ in $\xi$. %$\mathsf{h'_i\overset{\text{def}}{=}hash(h_i)(i \in [1,n])}$.
If a miner detects a fake key share, it would report it to $\xi$ within $T_1$. Specifically, the miner packages the detected fake key share as well as the signature of the key share message with $\mathsf{sn}$ in $\mathsf{Tx}_{\text{Appeal}}$, then sends it to $\xi$.
%
%$\mathsf{h'(k_{i})}$, the received fake key share $\mathsf{k'_{i}}$ with its owner's signature $\mathsf{Sig_{snd}}$ (i.e. the signature of $\mathcal{S}$ or $\mathcal{R}$ for $\mathsf{k'_{i}}$) in $\mathsf{Tx}_{\text{Appeal}}$, and sends it to $\xi$.
%
When $\xi$ gets $\mathsf{Tx}_{\text{Appeal}}$, it verifies the signature of the key share owner with $\mathsf{sn}$, and recalculates the hash of the reported fake key share to check if the data in $\mathsf{Tx}_{\text{Appeal}}$ is legitimate. 
If the verification succeeds, $\xi$ terminates the fair exchange protocol. Note that, $\mathsf{sn}$ can effectively prevent malicious behaviors of adversaries from destroying the execution of the protocol by providing the previous key share in $\mathsf{Tx}_{\text{Appeal}}$.
%$\mathsf{h'(k_{i})}$ based on $\mathsf{k'_{i}}$ and verifies $\mathsf{Sig_{snd}}$ to check if the data in $\mathsf{Tx}_{\text{Appeal}}$ is legitimate. If the verification succeeds, $\xi$ would terminate the fair exchange protocol.
%通过验证sig以及(t,n)VSS.verify来证明智能合约上的确实有错误。
\begin{equation}\notag
	\mathsf{Tx}_{\text{Appeal}} \overset{\text{def}}{=} (\mathsf{From: \mathcal{M}_i; To: \xi; Sig_{snd}, k'_{i}, sn), i \in [1,n]}.
\end{equation}

%$\mathsf{SN_{en}}$ is the encrypted random number based on the public key $\widetilde{\mathsf{pk}}_{V'}$ and the random number $\mathsf{N_{en}}$, and we can use the algorithm based on the Elliptic Curve Cryptography (ECC)~\cite{liu2008tinyecc} to implement that.
\noindent{\bf $\Theta.$Exchange.}
In this step, 
$\mathcal{S}$ first computes the ciphertext $\overline{\mathsf{m}}_\mathcal{S}$ of the exchange object $\mathsf{m_\mathcal{S}}$ based on $\mathsf{k_\mathcal{S}}$. This can be done by some common encryption technologies, e.g., Elliptic Curve Cryptography (ECC)~\cite{liu2008tinyecc} and MIMC~\cite{albrecht2016mimc}.
Then, $\mathcal{S}$ takes $\mathsf{pk}$, $\mathsf{k_\mathcal{S}}$ and $\mathsf{m_\mathcal{S}}$ to generate $\pi_\mathcal{S}$ based on the algorithm $\Pi.\mathsf{Prove}$, and sends ($\pi_\mathcal{S}$, $\overline{\mathsf{m}}_\mathcal{S}$, $\mathsf{h({m_\mathcal{S}})}$, $\mathsf{h({k_\mathcal{S}})}$) to $\mathcal{R}$.
Similarly, $\mathcal{R}$ sends ($\pi_\mathcal{R}$, $\mathsf{\overline{m}_\mathcal{R}}$, $\mathsf{h({m_\mathcal{R}})}$, $\mathsf{h({k_\mathcal{R}})}$) to $\mathcal{S}$.
After that, $\mathcal{S}$ and $\mathcal{R}$ use $\mathsf{vk}$ and the received data  
%($\pi$, $\mathsf{\overline{m}}$, $\mathsf{h({m})}$, $\mathsf{h({k})}$) 
to respectively verify ($\pi_\mathcal{S}$, $\pi_\mathcal{R}$) based on the algorithm $\Pi.\mathsf{Verify}$.
The above process is marked Step 2.1 and Step 2.3 in Fig.~\ref{Fig:fair-exchange}.
%\begin{equation}\notag
%\{1,0\} = \Pi.\mathsf{Verify}(\mathsf{vk, \pi, \overline{m}, \sigma_{m}, \mu_{k}}).
%\end{equation}

\noindent{\bf $\Theta.$Recover.} If both parties verify successfully, $\mathcal{S}$ and $\mathcal{R}$ send requests to the smart contract $\xi$ for key recovery (marked Step 3.1 to Step 3.7 in Fig.~\ref{Fig:fair-exchange}). 
When $\xi$ gets the requests from $\mathcal{S}$ and $\mathcal{R}$, it broadcasts this event to the selected miners, who then send their stored key shares to $\xi$ via $\mathsf{Tx}_{\text{Recover}}$.
\begin{equation}\notag
	\mathsf{Tx}_{\text{Recover}} \overset{\text{def}}{=} (\mathsf{From: \mathcal{M}_i; To: \xi; k_{\mathcal{S}i}, k_{\mathcal{R}i}), i\in[1,n]}.
\end{equation}
$\xi$ verifies the legitimacy of each key share by comparing its hash result and checking the address of its sender. When the number of valid key shares is greater than the key threshold $t$, $\xi$ sends the collected key shares to the requester, who then employs $\mathsf{(t,n)\text{-}VSS.Recover}$ to recover the key, and further decrypts the message.
%
%For example, $\mathcal{R}$ aims to get $\mathcal{S'}$ key $k_\mathcal{S}$. $\xi$ collects correct key shares and sends them to $\mathcal{R}$. When $\mathcal{R}$ gets the key share set, it executes the step in $\mathsf{(t,n)\text{-}VSS.Recover}$ to recover $k_\mathcal{S}$, and further uses $k_\mathcal{S}$ to get  ${m_\mathcal{S}}$ by decrypting $\overline{m_\mathcal{S}}$.
%Since miners follow a secure consensus algorithm to maintain the blockchain

%符号规定：从属的链用上标，例如A链的信息m表示为$m^A$;从属的用户用下标,例如用户S的信息m表示为$S_m$.       --Yihao 0913

\subsection{Cross-Channel and Applications} \label{cross-channel and applications}
Based on the hierarchical channel and the general fair exchange protocol presented in the previous two subsections, we present $\mathsf{Cross}$-$\mathsf{Channel}$ to support various cross-chain services, e.g., currency exchange (CE), fair exchange (FE, i.e., exchange currency with encrypted information), and encrypted information exchange (EIE)), in this subsection. We adopt the same notations as before: ($\mathcal{S^\alpha}$, $\mathcal{R^\alpha}$, $\mathcal{M^\alpha}$, $\xi^\alpha$) and ($\mathcal{S^\beta}$, $\mathcal{R^\beta}$, $\mathcal{M^\beta}$, $\xi^\beta$).
%Let $\alpha$ and $\beta$ be two blockchain systems, in which the accounts of $\mathcal{S}$ and $\mathcal{R}$, the miners, and the smart contracts, are denoted by ($\mathcal{S^\alpha}$, $\mathcal{R^\alpha}$, $\mathcal{M^\alpha}$, $\xi^\alpha$) and ($\mathcal{S^\beta}$, $\mathcal{R^\beta}$, $\mathcal{M^\beta}$, $\xi^\beta$), respectively.
Assume that $\mathcal{S}$ and $\mathcal{R}$ negotiate to exchange information ($\mathsf{m_{\mathcal{S}_i}, i \in \mathbb{Z^+}}$) in $\mathcal{S^\alpha}$ on blockchain $\alpha$ with information ($\mathsf{m_{\mathcal{R}_i}, i \in \mathbb{Z^+}}$) in $\mathcal{R^\beta}$ on blockchain $\beta$. 
%
%The accounts and smart contracts of $\alpha$ and $\beta$ are denoted as ($\mathcal{S^\alpha}$, $\mathcal{R^\alpha}$, $\mathcal{M^\alpha}$, $\xi^\alpha$) and ($\mathcal{S^\beta}$, $\mathcal{R^\beta}$, $\mathcal{M^\beta}$, $\xi^\beta$) respectively.
%
The whole scheme can be divided into four phases: {\bf Initialize}, {\bf Open}, {\bf Exchange}, and {\bf Close}, which are detailed in the following according to different application scenarios. Note that the first three phases are performed at each single chain while the last phase realizes the cross-chain operations via the cross-chain settlement protocol presented in Sec.~\ref{hierarchical-channel-scheme}. For better elaboration, we employ [ALL $\pmb{\Rightarrow}$] or [$\{\cdot\} \pmb{\Rightarrow}$] to denote that all or some of the three scenarios (CE, FE, EIE) need to execute the process that follows.
%Next, we introduce the specific process of our scheme in detail from different application scenarios.
%\subsubsection{Currency Exchange}

\noindent{\bf Initialize.} [ALL $\pmb{\Rightarrow}$] In $\alpha$ and $\beta$, each account is initialized with a unique address and a key pair ($\widetilde{\mathsf{pk}}$, $\widetilde{\mathsf{sk}}$). 
[(FE, EIE) $\pmb{\Rightarrow}$] Each digital commodity owner generates a common reference string $\mathsf{(pk, vk)}$ based on 
$\mathsf{\Theta.Setup}$ (introduced in Sec.~\ref{general-fair-exchange}).

\noindent{\bf Open.} [ALL $\pmb{\Rightarrow}$] According to the hierarchical interaction protocol $\Psi$, ($\mathcal{S^\alpha}$, $\mathcal{R^\alpha}$) and ($\mathcal{S^\beta}$, $\mathcal{R^\beta}$) respectively send $\mathsf{Tx_{open}}$ messages to call smart contracts $\xi^\alpha$ and $\xi^\beta$ to build channels $\Omega^\alpha_{0}$ and $\Omega^\beta_{0}$, and deposit their initial states, e.g., coins, into the channels. Note that the $\mathsf{Level}$ 0 of $\mathsf{Cross}$-$\mathsf{Channel}$ includes $\Omega^\alpha_0$ and $\Omega^\beta_0$, and the initial state is recorded as $[\mathcal{S}^\alpha \mapsto \mathsf{v}^\alpha_\mathcal{S}$, $\mathcal{R}^\alpha \mapsto \mathsf{v}^\alpha_\mathcal{R}$, $\mathcal{S}^\beta \mapsto \mathsf{v}^\beta_\mathcal{S}$, $\mathcal{R}^\beta \mapsto \mathsf{v}^\beta_\mathcal{R}]_{\Omega_0}$.
[EIE $\pmb{\Rightarrow}$] $\mathcal{S^\alpha}$ and $\mathcal{R^\beta}$ execute $\mathsf{\Theta.Share}$ to distribute their key shares.% and send $\mathsf{Tx}_{\text{Upload}}$. 

\noindent{\bf Exchange.} [CE $\pmb{\Rightarrow}$] The channel $\Omega_0$ allows two parties to instantaneously send payments between each other. 
[(FE,EIE) $\pmb{\Rightarrow}$] The sender implements $\mathsf{\Theta.Exchange}$ to encrypt the exchanged information, generates the zero-knowledge proof, and verifies the proof sent by the receiver.
For example, $\mathcal{S^\alpha}$ can generate multiple encrypted information $\mathsf{\overline{m}_{\mathcal{S}_i}}$ based on $\mathsf{m_{\mathcal{S}_i}}$ and $\mathsf{k_\mathcal{S}, i \in \mathbb{Z^+}}$, and send them in $\Omega^\alpha_0$ to $\mathcal{R^\alpha}$. 
Then, $\mathcal{S^\alpha}$ generates zero-knowledge proofs $\pi_{\mathcal{S}_i}$ to prove the authenticity of $\mathsf{k_\mathcal{S}}$ and $\mathsf{\overline{m}_{\mathcal{S}_i}}$ without exposing $\mathsf{k_\mathcal{S}}$ and $\mathsf{m_{\mathcal{S}_i}}$. 
Next, it sends $\pi_{\mathcal{S}_i}$ and the public parameters  shown in Fig.~\ref{Fig:fair-exchange-circuit2} to $\mathcal{R^\alpha}$.
$\mathcal{R^\alpha}$ can use $\mathsf{vk}$ and the received public parameters to verify $\pi_{\mathcal{S}_i}$.
Furthermore, with the consent of both parties in $\Omega_0$, one party can generate a sub-channel receipt $Sr$ to open a sub-channel and spend the unsettled amount based on the hierarchical interaction protocol $\Psi$ (details shown in Sec.~\ref{hierarchical-channel-scheme}). 

\noindent{\bf Close.} [ALL $\pmb{\Rightarrow}$] Based on the hierarchical settlement protocol $\Phi$,
%and the Algorithm~\ref{A:HCS}
all participants in the hierarchical channel are required to upload their final states based on $\mathsf{Tx}_{\text{Upload}}$ within $T_2$. 
[(CE) $\pmb{\Rightarrow}$] $\mathcal{S^\alpha}$ generates a preimage $\mathsf{pre}$ (a random 256-bits integer) and packages its hash result $\mathsf{h(pre)}$ in $\mathsf{Tx}_{\text{Lock}}$ to $\xi^\alpha$. %
\begin{equation}\notag
	\mathsf{Tx}_{\text{Lock}} \overset{\text{def}}{=} (\mathsf{From: \mathcal{S / R}; To: \xi; h(pre)}).
\end{equation}
$\xi^\alpha$ opens $\mathsf{h(pre)}$ and uses $\mathsf{h(pre)}$ to lock the state of $\mathcal{S^\alpha}$ in the channel $\Omega^\alpha_0$. When $\mathcal{R^\alpha}$ learns $\mathsf{h(pre)}$ in the blockchain $\alpha$, $\mathcal{R^\beta}$ sends $\mathsf{Tx}_{\text{Lock}}$ to lock the state of $\Omega^\beta_0$.  
We set timers $T_3$ and $T_4$ in blockchain $\alpha$ and $\beta$, respectively.
%We set two time thresholds $T_4$ and $T_5$ in blockchain $\alpha$.
One needs to provide $\mathsf{pre}$ within $T_3$ and $T_4$ to update the state of $\alpha$ and $\beta$, respectively. 
%By the same way,
%when one provides $\mathsf{pre}$ within $T_5$, $\xi^\beta$ updates the state of $\beta$. 
For example, $\mathcal{S^\beta}$ provides $\mathsf{pre}$ based on $\mathsf{Tx}_{\text{Update}}$ to update the states in blockchain $\beta$. Once $\mathsf{pre}$ is successfully verified by $\mathcal{\xi^\beta}$, $\mathcal{R^\beta}$ can learn $\mathsf{pre}$. Then, $\mathcal{R^\alpha}$ packages $\mathsf{pre}$ in $\mathsf{Tx}_{\text{Update}}$ to update the states of $\alpha$. 
Note that, according to the HTLC protocol, $T_4$ should be less than $T_3$, which effectively guarantees the atomicity of the interaction process (the related discussion is shown in Sec.~\ref{discussions}).
Moreover, considering the high latency of asynchronous networks, we set a time threshold $T_5$. When miners observe that the $T_3$ times out and no one upload $\mathsf{pre}$ in $\alpha$, they can offer $\mathsf{pre}$ within $T_5$ to get rewards. 
[(FE) $\pmb{\Rightarrow}$] Compared with the process in CE, $\mathcal{S^\alpha}$ needs to use $\mathsf{k_\mathcal{S}}$ as $\mathsf{pre}$ rather than regenerate a random 256-bit integer. 
%in which any miner can help $\mathcal{R^\alpha}$ provide $\mathsf{pre}$ to get rewards from $\mathcal{R^\alpha}$.

%将对异步的讨论转移至Sec.B hier. settlement协议中讨论       --Yihao 1001
% In the HTLC protocol, if no one submits $\mathsf{pre}$ before $T_3$ or $T_4$ times out, the smart contract would not update the interaction process within the channel to the blockchain. However, in an asynchronous network, affected by the high latency, after $\mathcal{S^\beta}$ provides $\mathsf{pre}$ to update the states in blockchain $\beta$, $\mathcal{R^\alpha}$ may fail to upload $\mathsf{pre}$ within $T_3$, which breaks the atomicity of HTLC. %and makes HTLC unsuitable for asynchronous networks. 
% Therefore, in order to make HTLC suitable for asynchronous networks, we set a timer $T_5$ in $\xi^\alpha$, in which any miner can help $\mathcal{R^\alpha}$ provide $\mathsf{pre}$ to get rewards from $\mathcal{R^\alpha}$.
%
\begin{equation}\notag
	\mathsf{Tx}_{\text{Update}} \overset{\text{def}}{=} (\mathsf{From: \mathcal{S / R / M}; To: \xi; pre}).
\end{equation}
%Like $T_3$, we set up a period time $T_5$ in which any entity can help $\mathcal{R_\alpha}$ provide $\mathsf{pre}$ to get rewards. 
% 
%[(FE) $\pmb{\Rightarrow}$] Unlike CE, $\mathcal{S_\alpha}$ and $\mathcal{R_\beta}$ can use $\mathsf{\mu_k}$ (shown in Fig.~\ref{Fig:fair-exchange-circuit}) as $\mathsf{h_{pre}}$ to lock the final states in $\xi$. 
%Then, $\mathcal{S_\beta}$ provides $\mathsf{k}$ to update the states in blockchain $\beta$. 
%
%After that, $\mathcal{R_\alpha}$ gets $\mathsf{k}$ and uses it to unlock the states in blockchain $\alpha$. 
%当有人提供时，可以实现状态更新。
[EIE $\pmb{\Rightarrow}$] In addition to $\mathsf{pre}$, $\mathcal{S^\alpha}$ and $\mathcal{R^\beta}$ also need to package the hash result of the key that needs to be recovered in $\mathsf{Tx}_{\text{Update-EIE}}$.
%to update the on-chain state and collect key shares. 
%
According to $\mathsf{\Theta.Recover}$, $\mathcal{S^\alpha}$ and $\mathcal{R^\beta}$ would get at least $t$ key shares from the smart contract, and they can execute $\mathsf{\Theta.Recover}$ to get each other's keys fairly ($\mathcal{S^\alpha}$ $\leftarrow$ $k_\mathcal{R}$, $\mathcal{R^\alpha}$ $\leftarrow$ $k_\mathcal{S}$).
%CE：pre, FE：$\mu_k$, EIE: $\mu_k$ + pre
\begin{equation}\notag
	\mathsf{Tx}_{\text{Update-EIE}} \overset{\text{def}}{=} (\mathsf{From: \mathcal{S / R / M}; To: \xi; pre, h(k)}).
\end{equation}
          
%TODO: 协议确定后，更改相应的协议图   --Yihao 2022 0915
%完成  --Yihao 2022 0924
For convenience, we summarize the logic of the smart contract $\xi^\alpha$ in Fig.~\ref{fig:contract}. Note that the smart contract $\xi^\beta$ is the same as $\xi^\alpha$ except that the timer $T_3$ is replaced with $T_4$. The entire protocol is outlined in Fig.~\ref{fig:Cross-Protocol}. 

\begin{figure}[htb]
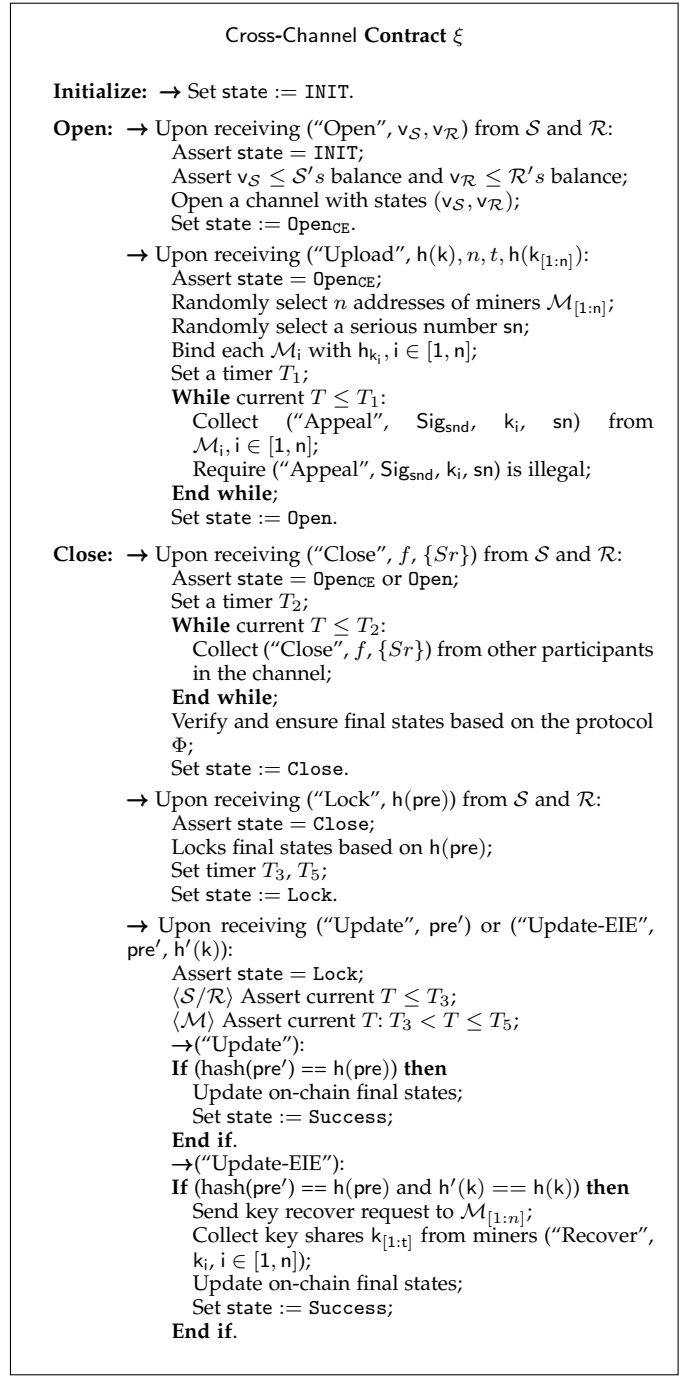
 \footnotesize %\scriptsize
	\caption{The $\mathsf{Cross}$-$\mathsf{Channel}$ Smart Contract $\xi^\alpha$. Suppose $\mathcal{S}$ and $\mathcal{R}$ are participants in the Level 0. Variables with brace represent a set (e.g., $Sr$ represents a sub-channel receipt, and $\{Sr\}$ is a set of multiple sub-channel receipts). 
	}\label{fig:contract}
	\begin{framed}
	%    \begin{multicols}{2}
		\centering \textbf{$\mathsf{Cross}$-$\mathsf{Channel}$ Contract $\xi$}\vspace{10pt}\\
			\begin{basedescript}{\desclabelwidth{30pt}}
				\item[\hspace{7pt}Initialize:]$\pmb{\rightarrow}$ Set $\mathsf{state} := \mathtt{INIT}$. \\
				\vspace{5pt}
				
				\item[\hspace{7pt}Open:]$\pmb{\rightarrow}$ Upon receiving (``Open", $\mathsf{v}_\mathcal{S}, \mathsf{v}_\mathcal{R}$) from $\mathcal{S}$ and $ \mathcal{R}$:\\
				{
					\setlength{\leftskip}{17pt}Assert $\mathsf{state} = \mathtt{INIT}$;\\
					
					\setlength{\leftskip}{17pt}Assert $\mathsf{v}_\mathcal{S} \leq \mathcal{S}'s$ balance and $\mathsf{v}_\mathcal{R} \leq \mathcal{R}'s$ balance;\\
					
					\setlength{\leftskip}{17pt}Open a channel with states $(\mathsf{v}_\mathcal{S}, \mathsf{v}_\mathcal{R})$; \\
					
					\setlength{\leftskip}{17pt}Set $\mathsf{state} := \mathtt{Open_{CE}}$. \\
					
					\vspace{2.5pt}
					
					\setlength{\leftskip}{0pt}$\pmb{\rightarrow}$ Upon receiving (``Upload", $\mathsf{h(k}), n, t, \mathsf{h(k_{[1:n]}})$: \\

					\setlength{\leftskip}{17pt}Assert $\mathsf{state} = \mathtt{Open_{CE}}$;\\
					
					\setlength{\leftskip}{17pt}Randomly select $n$ addresses of miners $\mathsf{\mathcal{M}_{[1:n]}}$;\\
					
					\setlength{\leftskip}{17pt}Randomly select a serious number $\mathsf{sn}$;\\
					
					\setlength{\leftskip}{17pt}Bind each $\mathsf{\mathcal{M}_{i}}$ with $\mathsf{h_{k_{i}}, i\in[1,n]}$; \\
					
					\setlength{\leftskip}{17pt}Set a timer $T_1$; \\
					
					\setlength{\leftskip}{17pt}{\bf While} current $T \le T_1$: \\
					
					\setlength{\leftskip}{25pt}Collect (``Appeal", $\mathsf{Sig_{snd}}$, $\mathsf{{k_{i}}}$, $\mathsf{sn}$) from $\mathsf{\mathcal{M}_{i}, i \in [1,n]}$; \\
					
					\setlength{\leftskip}{25pt}Require (``Appeal", $\mathsf{Sig_{snd}}$, $\mathsf{{k_{i}}}$, $\mathsf{sn}$) is illegal; \\ % (违规返回1)

					\setlength{\leftskip}{17pt}{\bf End while};\\

					\setlength{\leftskip}{17pt}Set $\mathsf{state} := \mathtt{Open}$.\\
				}
				
				\vspace{5pt}
				
				\item[\hspace{7pt}Close:]$\pmb{\rightarrow}$ Upon receiving (``Close", $f$, \{$Sr$\}) from $\mathcal{S}$ and $\mathcal{R}$:\\
				{
					\setlength{\leftskip}{17pt}Assert $\mathsf{state} = \mathtt{Open_{CE}}$ or $\mathtt{Open}$; \\
					
					\setlength{\leftskip}{17pt}Set a timer $T_2$; \\
					
					\setlength{\leftskip}{17pt}{\bf While} current $T \le T_2$: \\
					
					\setlength{\leftskip}{25pt}Collect (``Close", $f$, \{$Sr$\}) from other participants in the channel; \\
					
					\setlength{\leftskip}{17pt}{\bf End while};\\

					\setlength{\leftskip}{17pt}Verify and ensure final states based on the protocol $\Phi$; \\
					
					\setlength{\leftskip}{17pt}Set $\mathsf{state} := \mathtt{Close}$.\\
					
					\vspace{2.5pt}
					
					\setlength{\leftskip}{0pt}$\pmb{\rightarrow}$ Upon receiving  (``Lock", $\mathsf{h(pre)}$) 
					from $\mathcal{S}$ and $\mathcal{R}$:\\
					
					\setlength{\leftskip}{17pt}Assert $\mathsf{state} = \mathtt{Close}$; \\
					
					\setlength{\leftskip}{17pt}Locks final states based on $\mathsf{h(pre)}$; \\
					
					\setlength{\leftskip}{17pt}Set timer $T_{3}$, $T_5$; \\
					
					\setlength{\leftskip}{17pt}Set $\mathsf{state} := \mathtt{Lock}$. \\
					
					\vspace{2.5pt}		
					
					\setlength{\leftskip}{0pt}$\pmb{\rightarrow}$ Upon receiving (``Update", $\mathsf{pre'}$) or (``Update-EIE", $\mathsf{pre'}$, $\mathsf{h'(k)}$):\\
					
					\setlength{\leftskip}{17pt}Assert $\mathsf{state} = \mathtt{Lock}$; \\
					
					\setlength{\leftskip}{17pt}$\left \langle \mathcal{S}/ \mathcal{R} \right \rangle$ Assert current $T \le T_{3}$; \\
					
					%	$\left \langle a,b \right \rangle$
					
					\setlength{\leftskip}{17pt}$\left \langle \mathcal{M} \right \rangle$ Assert current $T$: $T_{3} < T \le T_5$; \\
					
					\setlength{\leftskip}{17pt}$\pmb{\rightarrow}$(``Update"): \\
					
					\setlength{\leftskip}{17pt}{\bf If} (hash($\mathsf{pre'}$) == $\mathsf{h({pre})}$)  {\bf then} \\
					
					\setlength{\leftskip}{25pt}Update on-chain final states;\\
					
					\setlength{\leftskip}{25pt}Set $\mathsf{state} := \mathtt{Success}$; \\
					
					\setlength{\leftskip}{17pt}{\bf End if}. \\
					
					\setlength{\leftskip}{17pt}$\pmb{\rightarrow}$(``Update-EIE"): \\
					
					\setlength{\leftskip}{17pt}{\bf If} (hash($\mathsf{pre'}$) == $\mathsf{h({pre})}$ and $\mathsf{h'(k}) == \mathsf{h(k)}$)  {\bf then} \\
					
					\setlength{\leftskip}{25pt}Send key recover request to $\mathcal{M}_{[1:n]}$;\\
					
					\setlength{\leftskip}{25pt}Collect key shares $\mathsf{k_{[1:t]}}$ from miners (``Recover", $\mathsf{k_{i}}$, $\mathsf{i\in[1,n]}$); \\
					%因为智能合约本身时公开的，所以只需要收集，用户自己需要时恢复即可。
					
					\setlength{\leftskip}{25pt}Update on-chain final states;\\
					
					\setlength{\leftskip}{25pt}Set $\mathsf{state} := \mathtt{Success}$; \\
					
					\setlength{\leftskip}{17pt}{\bf End if}. \\
					
					\vspace{2.5pt}
				}
			\end{basedescript}
	%		\end{multicols}\vspace{-15pt}
	\end{framed}\vspace{-10pt}
\end{figure}
\begin{figure*}[htb]
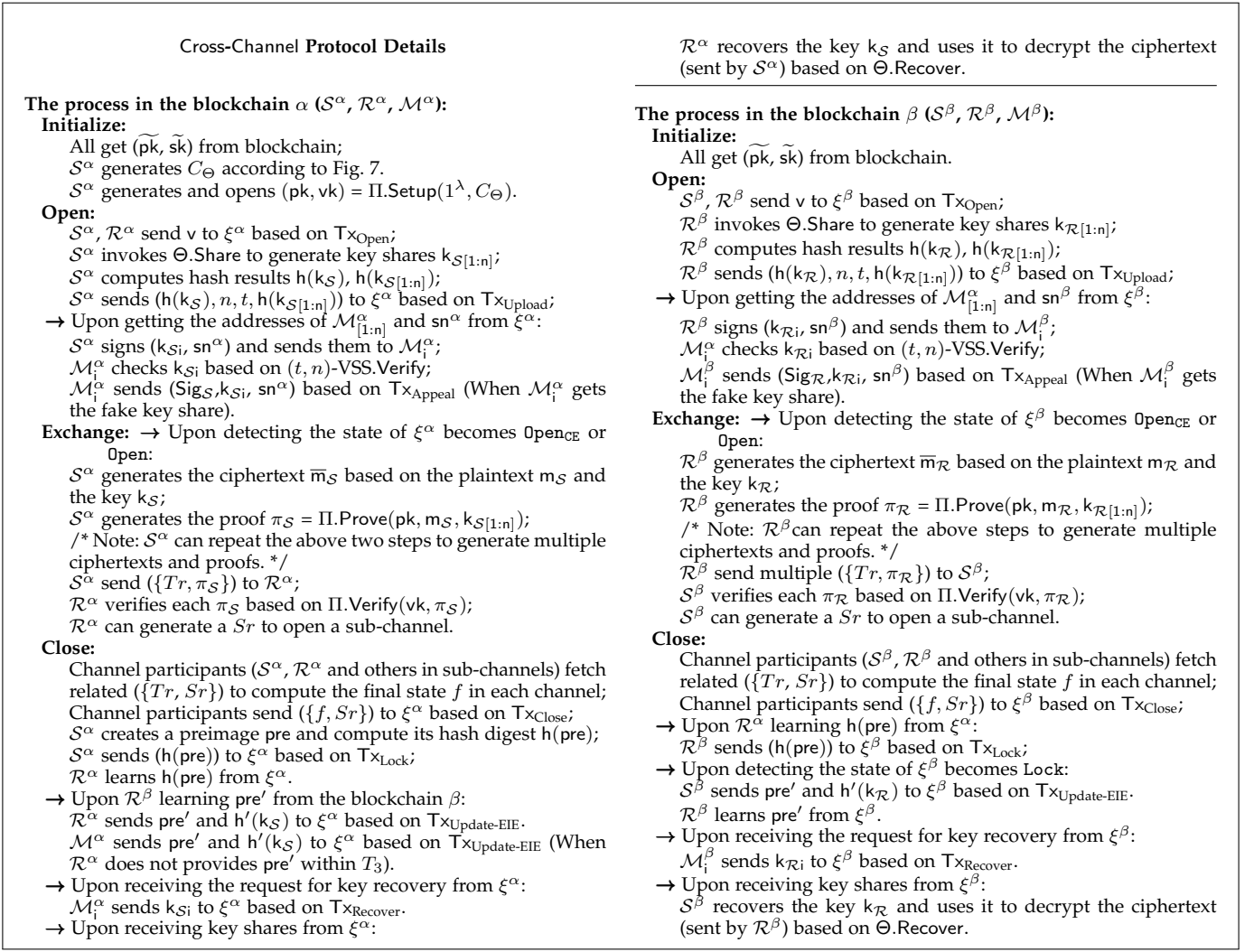
 \footnotesize%\scriptsize
	\caption{An example of the encrypted information exchange (EIE) based on $\mathsf{Cross}$-$\mathsf{Channel}$. Variables with brace represent a set (e.g., $Tr$ represents a transaction receipt, and $\{Tr\}$ is a set of multiple transaction receipts).} \label{fig:Cross-Protocol}
	\begin{framed}
		\begin{multicols}{2}
			\centering \textbf{$\mathsf{Cross}$-$\mathsf{Channel}$ Protocol Details} \vspace{5pt}
			\begin{basedescript}{\desclabelwidth{30pt}}
				%%%%%%%%%%%%%%%%%%%%%%%%%%%%%%%%%%%%%%%%%
				\item[{\bf The process in the blockchain $\alpha$ ($\mathcal{S^\alpha}$, $\mathcal{R^\alpha}$, $\mathcal{M^\alpha}$):}] 
				\item[\hspace{7pt}Initialize:]~\\
				{
					\setlength{\leftskip}{-16pt}All get ($\widetilde{\mathsf{pk}}$, $\widetilde{\mathsf{sk}}$) from blockchain; \\ 
					
					\setlength{\leftskip}{-16pt}$\mathcal{S^\alpha}$ generates $C_\Theta$ according to Fig.~\ref{Fig:fair-exchange-circuit2}. \\
					
					\setlength{\leftskip}{-16pt}$\mathcal{S^\alpha}$ generates and opens $\mathsf{(pk,vk)}$ = $\Pi$.${\mathsf{Setup}}(1^{\lambda}, C_\Theta)$.\\
				}
				\item[\hspace{7pt}Open:] ~\\ 
				{
					\setlength{\leftskip}{-16pt}$\mathcal{S^\alpha}$, $\mathcal{R^\alpha}$ send $\mathsf{v}$ to $\xi^\alpha$ based on $\mathsf{Tx}_{\text{Open}}$;\\
					
					\setlength{\leftskip}{-16pt}$\mathcal{S^\alpha}$ invokes $\mathsf{\Theta.Share}$ to generate key shares $\mathsf{k_{\mathcal{S}[1:n]}}$;\\
					
					\setlength{\leftskip}{-16pt}$\mathcal{S^\alpha}$ computes hash results $\mathsf{h(k_\mathcal{S})}$, $\mathsf{h(k_{\mathcal{S}[1:n]})}$;\\
					
					\setlength{\leftskip}{-16pt}$\mathcal{S^\alpha}$ sends ($\mathsf{h(k_\mathcal{S})}, n, t, \mathsf{h(k_{\mathcal{S}[1:n]})}$) to $\xi^\alpha$ based on $\mathsf{Tx}_{\text{Upload}}$;\\
					
					\setlength{\leftskip}{-26pt}$\pmb{\rightarrow}$ Upon getting the addresses of $\mathsf{\mathcal{M}^\alpha_{[1:n]}}$ and $\mathsf{sn}^\alpha$ from $\xi^\alpha$:\\
					
					\setlength{\leftskip}{-16pt}$\mathcal{S^\alpha}$ signs ($\mathsf{k_{\mathcal{S}i}}$, $\mathsf{sn}^\alpha$) and sends them to $\mathsf{\mathcal{M}^\alpha_{i}}$;\\
					
		        	\setlength{\leftskip}{-16pt}$\mathsf{\mathcal{M}^\alpha_{i}}$ checks $\mathsf{k_{\mathcal{S}i}}$ based on $(t,n)$-VSS.$\mathsf{Verify}$;\\
		        	
					\setlength{\leftskip}{-16pt}$\mathsf{\mathcal{M}^\alpha_{i}}$ sends ($\mathsf{Sig_\mathcal{S}}$,$\mathsf{{k_{\mathcal{S}i}}}$, $\mathsf{sn}^\alpha$) based on $\mathsf{Tx}_{\text{Appeal}}$ (When $\mathsf{\mathcal{M}^\alpha_{i}}$ gets the fake key share).\\
				}
				\item[\hspace{7pt}Exchange:]$\pmb{\rightarrow}$ Upon detecting the state of $\xi^\alpha$ becomes $\mathtt{Open_{CE}}$ or $\mathtt{Open}$:\\
				{
					\setlength{\leftskip}{-16pt}$\mathcal{S^\alpha}$ generates the ciphertext $\mathsf{\overline{m}_{{\mathcal{S}}}}$ based on the plaintext $\mathsf{m_{\mathcal{S}}}$ and the key $\mathsf{k_\mathcal{S}}$;\\
					
					\setlength{\leftskip}{-16pt}$\mathcal{S^\alpha}$ generates the proof $\pi_{\mathcal{S}}$ = $\Pi.\mathsf{Prove}(\mathsf{pk, m_{\mathcal{S}}, k_{\mathcal{S}[1:n]})}$; \\
					
					\setlength{\leftskip}{-16pt}/* Note: $\mathcal{S^\alpha}$ can repeat the above two steps to generate multiple ciphertexts and proofs. */ \\
					
					\setlength{\leftskip}{-16pt}$\mathcal{S^\alpha}$ send (\{$Tr, \pi_\mathcal{S}$\}) to $\mathcal{R}^\alpha$;\\
					
					\setlength{\leftskip}{-16pt}$\mathcal{R^\alpha}$ verifies each $\pi_{\mathcal{S}}$ based on $\Pi.\mathsf{Verify(vk,\pi_{\mathcal{S}})}$;\\
					
					\setlength{\leftskip}{-16pt}$\mathcal{R^\alpha}$ can generate a $Sr$ to open a sub-channel. \\
				}
				\item[\hspace{7pt}Close:]~\\
				{
					\setlength{\leftskip}{-16pt}Channel participants ($\mathcal{S^\alpha}$, $\mathcal{R^\alpha}$ and others in sub-channels) fetch related (\{$Tr$, $Sr$\}) to compute the final state $f$ in each channel;\\
					
					\setlength{\leftskip}{-16pt}Channel participants send (\{$f, Sr$\}) to $\xi^\alpha$ based on $\mathsf{Tx}_{\text{Close}}$;\\
					
					\setlength{\leftskip}{-16pt}$\mathcal{S^\alpha}$ creates a preimage $\mathsf{pre}$ and compute its hash digest $\mathsf{h({pre})}$;\\
					
					\setlength{\leftskip}{-16pt}$\mathcal{S^\alpha}$ sends ($\mathsf{h({pre})}$) to $\xi^\alpha$ based on $\mathsf{Tx}_{\text{Lock}}$;\\
					
					\setlength{\leftskip}{-16pt}$\mathcal{R^\alpha}$ learns $\mathsf{h({pre})}$ from $\xi^\alpha$.\\
					
                    \setlength{\leftskip}{-26pt}$\pmb{\rightarrow}$ Upon $\mathcal{R^\beta}$ learning $\mathsf{pre'}$ from the blockchain $\beta$: \\
                    
                    \setlength{\leftskip}{-16pt}$\mathcal{R^\alpha}$ sends $\mathsf{pre'}$ and $\mathsf{h'(k_\mathcal{S})}$ to $\xi^\alpha$ based on $\mathsf{Tx}_{\text{Update-EIE}}$.\\
                    
                    \setlength{\leftskip}{-16pt}$\mathcal{M^\alpha}$ sends $\mathsf{pre'}$ and $\mathsf{h'(k_\mathcal{S})}$ to $\xi^\alpha$ based on $\mathsf{Tx}_{\text{Update-EIE}}$ (When $\mathcal{R^\alpha}$ does not provides $\mathsf{pre'}$ within $T_3$).\\
                    
                    \setlength{\leftskip}{-26pt}$\pmb{\rightarrow}$ Upon receiving the request for key recovery from $\xi^\alpha$:\\
                    
                    \setlength{\leftskip}{-16pt}$\mathsf{\mathcal{M}_{i}^\alpha}$ sends $\mathsf{k_{\mathcal{S}i}}$ to $\xi^\alpha$ based on $\mathsf{Tx}_{\text{Recover}}$.\\
                    
                    \setlength{\leftskip}{-26pt}$\pmb{\rightarrow}$ Upon receiving key shares from $\xi^\alpha$:\\
                    
                    \setlength{\leftskip}{-16pt} $\mathcal{R^\alpha}$ recovers the key $\mathsf{k_\mathcal{S}}$ and uses it to decrypt the ciphertext (sent by $\mathcal{S}^\alpha$) based on $\mathsf{\Theta.Recover}$.\\
				}
			\end{basedescript}
			
			\hrule
			\vspace{5pt}

			\begin{basedescript}{\desclabelwidth{30pt}}
				\item[{\bf The process in the blockchain $\beta$ ($\mathcal{S^\beta}$, $\mathcal{R^\beta}$, $\mathcal{M^\beta}$):}]
				\item[\hspace{7pt}Initialize:]~\\
				{
					\setlength{\leftskip}{-16pt}All get ($\widetilde{\mathsf{pk}}$, $\widetilde{\mathsf{sk}}$) from blockchain. \\ 
				}
				\item[\hspace{7pt}Open:]~\\  
				{
					\setlength{\leftskip}{-16pt}$\mathcal{S^\beta}$, $\mathcal{R^\beta}$ send $\mathsf{v}$ to $\xi^\beta$ based on $\mathsf{Tx}_{\text{Open}}$;\\
					
					\setlength{\leftskip}{-16pt}$\mathcal{R^\beta}$ invokes $\mathsf{\Theta.Share}$ to generate key shares $\mathsf{k_{\mathcal{R}[1:n]}}$;\\
					
					\setlength{\leftskip}{-16pt}$\mathcal{R^\beta}$ computes hash results $\mathsf{h(k_\mathcal{R})}$, $\mathsf{h(k_{\mathcal{R}[1:n]})}$;\\
					
					\setlength{\leftskip}{-16pt}$\mathcal{R^\beta}$ sends ($\mathsf{h(k_\mathcal{R})}, n, t, \mathsf{h(k_{\mathcal{R}[1:n]})}$) to $\xi^\beta$ based on $\mathsf{Tx}_{\text{Upload}}$;\\
					
					\setlength{\leftskip}{-26pt}$\pmb{\rightarrow}$ Upon getting the addresses of $\mathsf{\mathcal{M}^\alpha_{[1:n]}}$ and $\mathsf{sn}^\beta$ from $\xi^\beta$:\\
					
					\setlength{\leftskip}{-16pt}$\mathcal{R^\beta}$ signs ($\mathsf{k_{\mathcal{R}i}}$, $\mathsf{sn}^\beta$) and sends them to $\mathsf{\mathcal{M}^\beta_{i}}$;\\
					
		        	\setlength{\leftskip}{-16pt}$\mathsf{\mathcal{M}^\alpha_{i}}$ checks $\mathsf{k_{\mathcal{R}i}}$ based on $(t,n)$-VSS.$\mathsf{Verify}$;\\
		        	
					\setlength{\leftskip}{-16pt}$\mathsf{\mathcal{M}^\beta_{i}}$ sends ($\mathsf{Sig_\mathcal{R}}$,$\mathsf{{k_{\mathcal{R}i}}}$, $\mathsf{sn}^\beta$) based on $\mathsf{Tx}_{\text{Appeal}}$ (When $\mathsf{\mathcal{M}^\beta_{i}}$ gets the fake key share).\\
				}
				\item[\hspace{7pt}Exchange:]$\pmb{\rightarrow}$ Upon detecting the state of $\xi^\beta$ becomes $\mathtt{Open_{CE}}$ or $\mathtt{Open}$:\\
				{
					\setlength{\leftskip}{-16pt}$\mathcal{R^\beta}$ generates the ciphertext $\mathsf{\overline{m}_{{\mathcal{R}}}}$ based on the plaintext $\mathsf{m_{\mathcal{R}}}$ and the key $\mathsf{k_\mathcal{R}}$;\\
					
					\setlength{\leftskip}{-16pt}$\mathcal{R^\beta}$ generates the proof $\pi_{\mathcal{R}}$ = $\Pi.\mathsf{Prove}(\mathsf{pk, m_{\mathcal{R}}, k_{\mathcal{R}[1:n]}})$; \\
					
					\setlength{\leftskip}{-16pt}/* Note: $\mathcal{R^\beta}$can repeat the above steps to generate multiple ciphertexts and proofs. */ \\
					
					\setlength{\leftskip}{-16pt}$\mathcal{R^\beta}$ send multiple (\{$Tr, \pi_\mathcal{R}$\}) to $\mathcal{S}^\beta$;\\
					
					\setlength{\leftskip}{-16pt}$\mathcal{S^\beta}$ verifies each $\pi_{\mathcal{R}}$ based on $\Pi.\mathsf{Verify(vk, \pi_{\mathcal{R}})}$; \\
					
					\setlength{\leftskip}{-16pt}$\mathcal{S^\beta}$ can generate a $Sr$ to open a sub-channel. \\
				}
                \item[\hspace{7pt}Close:]~\\
				{
					\setlength{\leftskip}{-16pt}Channel participants ($\mathcal{S^\beta}$, $\mathcal{R^\beta}$ and others in sub-channels) fetch related (\{$Tr$, $Sr$\}) to compute the final state $f$ in each channel;\\
					
					\setlength{\leftskip}{-16pt}Channel participants send (\{$f, Sr$\}) to $\xi^\beta$ based on $\mathsf{Tx}_{\text{Close}}$;\\
					
					\setlength{\leftskip}{-26pt}$\pmb{\rightarrow}$ Upon $\mathcal{R^\alpha}$ learning $\mathsf{h({pre})}$ from $\xi^\alpha$:\\
					
					\setlength{\leftskip}{-16pt}$\mathcal{R^\beta}$ sends ($\mathsf{h({pre})}$) to $\xi^\beta$ based on $\mathsf{Tx}_{\text{Lock}}$;\\
					
					\setlength{\leftskip}{-26pt}$\pmb{\rightarrow}$ Upon detecting the state of $\xi^\beta$ becomes $\mathtt{Lock}$: \\
					
					\setlength{\leftskip}{-16pt}$\mathcal{S^\beta}$ sends $\mathsf{pre'}$ and $\mathsf{h'(k_\mathcal{R})}$ to $\xi^\beta$ based on $\mathsf{Tx}_{\text{Update-EIE}}$.\\
					
					\setlength{\leftskip}{-16pt}$\mathcal{R^\beta}$ learns $\mathsf{{pre'}}$ from $\xi^\beta$.\\
                    \setlength{\leftskip}{-26pt}$\pmb{\rightarrow}$ Upon receiving the request for key recovery from $\xi^\beta$: \\
                    
                    \setlength{\leftskip}{-16pt}$\mathsf{\mathcal{M}_{i}^\beta}$ sends $\mathsf{k_{\mathcal{R}i}}$ to $\xi^\beta$ based on $\mathsf{Tx}_{\text{Recover}}$.\\
                    
                    \setlength{\leftskip}{-26pt}$\pmb{\rightarrow}$ Upon receiving key shares from $\xi^\beta$:\\
                    
                    \setlength{\leftskip}{-16pt} $\mathcal{S^\beta}$ recovers the key $\mathsf{k_\mathcal{R}}$ and uses it to decrypt the ciphertext (sent by $\mathcal{R}^\beta$) based on $\mathsf{\Theta.Recover}$.\\
				}
			\end{basedescript}
		\end{multicols}\vspace{-15pt}
	\end{framed}\vspace{-10pt}
\end{figure*}

%TODO:原子性分析部分需要重新该写   --Yihao  0919
%完成    --Yihao  0924
\subsection{Analysis} \label{discussions}
In this subsection, we %informally
prove that $\mathsf{Cross}$-$\mathsf{Channel}$ possesses the properties of fairness and atomicity. Its scalability will be demonstrated through experiments in Sec.~\ref{Performance}. %We give the assumption as follows.
%可以兼容现有的通道方案,对目前的单层通道做改进
%zk的prover过长的问题可以避免
%只需setup一次
%To deduce the Theorem~\ref{theorem-1} for scalability, we first give the Assumption~\ref{assumption-1} and prove the correctness of Lemma~\ref{lem-1}.
\begin{assumption}
	We assume that behaviors of the participants $\mathcal{S, R, M}$ within each chain follow the security model defined in Sec.~\ref{model}. 
	Consider Byzantine fault-tolerance, we also assume that an $N_{node}$-node %($3\ell+1$)-node 
 network can tolerate up to $\ell$ Byzantine nodes and support any reconstruction threshold within $ [\ell+1, N_{node}-\ell]$~\cite{das2022practical,huang2022workload}, where $N_{node}=3\ell+1$. 
	\label{assumption-1}
\end{assumption}
\begin{lemma} 
	zk-SNARK is a non-interactive zero-knowledge proof technique that satisfies completeness, soundness, and zero knowledge~\cite{groth2016size}.
	\label{lem-1}
\end{lemma}
\begin{lemma}
	Based on zk-SNARK, a sender (e.g., $\mathcal{S}$) can convince a receiver (e.g., $\mathcal{R}$) that the output ciphertext $\overline{m}$ is encrypted based on $m$ and $k_{[1:n]}$. Besides that, the receiver $\mathcal{R}$ cannot get any knowledge about $m$ and $k_{[1:n]}$ from public information.
	\label{lem-2}
\end{lemma}
\begin{proof}
    Based on our model defined in Sec.~\ref{model}, $\mathcal{S}$ and $\mathcal{R}$ can be arbitrarily malicious. So,  
	a malicious $\mathcal{S}$ would provide fake $m$ and $k_{[1:n]}$ to convince an honest $\mathcal{R}$, and a malicious $\mathcal{R}$ would try to obtain $m$ and $k_{[1:n]}$ from the public information provided by $\mathcal{S}$. 
	In order to prevent the above malicious behaviors, we adopt zk-SNARK (introduced in Sec.~\ref{zk-snark}), and design a new circuit (introduced in Fig.~\ref{Fig:fair-exchange-circuit2}) for it.  Specifically, we take the encrypted information $m$ and key shares $k_{[1:n]}$ as private inputs, and implement the Encrypt, Recover, and Hash algorithms in the circuit. 
	Based on Lemma~\ref{lem-1}, the completeness of zk-SNARK ensures that an honest $\mathcal{S}$ with valid $m$ and $k_{[1:n]}$ can always convince $\mathcal{R}$. 
	When $\mathcal{S}$ is malicious, the soundness of zk-SNARK makes it impossible for $\mathcal{S}$ (with probabilistic polynomial-time witness extractor $\mathcal{E}$) to provide fake secrets (i.e., $m$ and $k_{[1:n]}$) to deceive $\mathcal{R}$ (shown in Equation~(\ref{eq:soundness})). In other words, $\mathcal{R}$ can determine whether $\mathcal{S}$ provides fake private inputs based on the public parameters, e.g., the common reference string $\mathsf{crs} (\mathsf{pk,vk})$, the proof $\pi$, and the public inputs. 
	\begin{equation}  % \scriptsize
		\Pr{\left[
			\begin{split}
				& C(m, k) \neq R_C \\
				& {\mathrm{Verify}}(\mathsf{vk}, \pi) = 1
			\end{split}
			\Biggm\vert
			\begin{split}
				& {\mathtt{Setup}}(1^{\lambda}, C) \rightarrow (\mathsf{crs})\\
				& \mathcal{S}(\mathsf{pk, vk}) \rightarrow (\pi)\\
				& \mathcal{E}(\mathsf{pk, vk}) \rightarrow (m, k_{[1:n]})
			\end{split}
			\right]} \le 
			\mathsf{negl}(\lambda).
		\label{eq:soundness}
	\end{equation}
	The zero knowledge of zk-SNARK ensures that for every probabilistic polynomial-time (PPT) malicious $\mathcal{R}$, the probability that $\mathcal{R}$ can take private inputs from the proof can be ignored. 
	%In other word, 
    %Suppose there is a simulator that can extract a witness from a proof even after it has seen many simulated proofs.
	%for every $\mathcal{A}$, there is a probabilistic polynomial-time witness extractor $\mathcal{E}$, such that:
	Thus one can conclude that zk-SNARK in our general fair exchange protocol ensures the authenticity and privacy of plaintexts and key shares. %Next, based on Lemma~\ref{lem-2} and Lemma~\ref{lem-3}, we prove the correctness of Theorem~\ref{theorem-3}.
\end{proof}
\begin{theorem}[Fairness]
	In the scenario of encrypted information exchange,  
	$\mathcal{S}$ and $\mathcal{R}$ can be guaranteed that if $\mathcal{S}$ gets $\mathcal{R}$'s secret (i.e., the plaintext), $\mathcal{R}$ would also get $\mathcal{S}$'s secret, and vice versa. 
	\label{theorem-3}
\end{theorem}
\begin{proof}
    We design the general fair exchange protocol $\Theta$ to achieve fairness in $\mathsf{Cross}$-$\mathsf{Channel}$.
    %The general fair exchange protocol adopts zk-SNARK and $(t,n)$-VSS to ensure the fairness of the interaction. 
    In the process of {\bf $\Theta$.Exchange} (details in  Sec.~\ref{general-fair-exchange}), based on Lemma~\ref{lem-2}, one can get that 
    the general fair exchange protocol adopts zk-SNARK to help $\mathcal{S}$ and $\mathcal{R}$ achieve the authenticity and privacy of the exchanged information. 
    So, in this step, neither side has access to the other's secrets.
    In other steps, 
	the general fair exchange protocol adopts $(t,n)$-VSS.  Let's consider the impact of Byzantine fault-tolerance on the fairness of the protocol. %We assume a worst scenario where all Byzantine nodes are selected as the recipients of the key share. Under this assumption, 
	%Based on Assumption~\ref{assumption-1} ($N_{node}=3\ell+1$), we 
 Suppose that $n$ nodes have been chosen to receive the key shares and the key threshold is $t$. There are two possible Byzantine behaviors that can break the protocol. 
	First, when the number of Byzantine nodes receiving the key shares is greater than or equal to $t$, these nodes can collude to break the fairness and recover the key. 
	Second, when the key threshold $t$ is smaller than the number of honest nodes in the $n$ nodes receiving the key shares, the sender of {\bf $\Theta$.Recover} cannot recover the key when all the Byzantine nodes maliciously refuse to provide their key shares because the number of key shares is less than $t$. To overcome the first problem, we require that $t>\ell$, and to counter the second one, we set $n\ge t+\ell$, where $\ell$ is the maximum number of Byzantine nodes in the whole network. These two constraints can accommodate the worst case in which all the $\ell$ malicious nodes are unluckily selected to receive the key shares.
 %
 %greater than the number of honest nodes $(n-\ell^n)$, where $\ell^n$ represents the upper bound of the Byzantine nodes among the $n$ nodes, the Byzantine nodes would maliciously refuse to provide their key shares, and then the sender of {\bf $\Theta$.Recover} cannot recover the key because the number of key shares is less than $t$. 
	%
%	To avoid the above two problems, we set that the threshold $(t,n)$ meets the following constraints: $\{n = N_{node}$, $t >\ell\}$, to consider the worst case in which all the $\ell$ malicious nodes are unluckily selected to receive the key shares.
% $\{n = 3\ell^n+1$, $t \in [\ell^n+1, n-\ell^n]\}$. 
%	After many repetitions of randomly selecting $n$ nodes, we get the expected number of Byzantine nodes among $n$ nodes is $n/3$, and $\ell^n$=$\ell$ in the worst case. 
%	Therefore, to ensure the security, we need to consider the value of $(n,t)$ that is Byzantine fault-tolerant in the worst case, i.e. $\{n = 3\ell+1$=$N_{node}$, $t \in [(N_{node}+2)/3, (2*N_{node}+1)/3]\}$. 
	%So, we set that the threshold $(t,n)$ meets the constraints: $\{n = 3\ell^n+1$, $t \in [\ell^n+1, n-\ell^n]\}$. 
	In summary, one can see that by carefully setting $n$ and $t$ it is impossible for the Byzantine nodes to collect all $t$ key shares even if all Byzantine nodes collude, and thus allow the key to be successfully recovered even if all Byzantine nodes do not follow the protocol to send $\mathsf{Tx}_{\text{Recover}}$. By this way one can guarantee that the general fair exchange protocol can effectively provide fairness for the interaction between two parties.
%	 we set $\{(n-f^n_B) \geq t > f^n_B\}$, where $f^n_B$ is the number of Byzantine nodes among the selected $n$ nodes. Specifically, in the first case, $\{t > f^n_B\}$ makes it impossible for Byzantine nodes to collect $t$ key shares even if all Byzantine nodes $f^n_B$ colluded. As for the second case, we set $\{t \leq (n-f_B)\}$, which allows both sides of the cross-chain interaction to still successfully recover the key if all Byzantine nodes $f^n_B$ do not follow the protocol to send $\mathsf{Tx}_{\text{Recover}}$. 
%	By repeating the experiment many times, we get that the expectation of the number of Byzantine nodes in randomly selected $n$ nodes is $n$/3. So, we can further select the appropriate values of $t$ and $n$ as needed, as long as the relationship between the two is consistent with $\{(2*n)/3 \geq t > n/3\}$ (details shown in Sec.\ref{Performance}). 
\end{proof}
\begin{theorem}[Atomicity]
	Let objects $x_\mathcal{S}$ $\in$ $\mathcal{S}$ and $x_\mathcal{R}$ $\in$ $\mathcal{R}$ before a cross-chain exchange. The settlement result after $\mathsf{Cross}$-$\mathsf{Channel}$ can only be $(x_\mathcal{S}$ $\in$ $\mathcal{S}$ $\bigwedge$ $x_\mathcal{R}$ $\in$ $\mathcal{R})$ $\bigvee$ $(x_\mathcal{S}$ $\in$ $\mathcal{R}$ $\bigwedge$ $x_\mathcal{R}$ $\in$ $\mathcal{S})$.%, even the network is asynchronous.
	\label{theorem-4}
\end{theorem}
\begin{proof}
	Before closing the channel, 
	the refusal of settlement by one or both parties would result in $(x_\mathcal{S}$ $\in$ $\mathcal{S}$ $\bigwedge$ $x_\mathcal{R}$ $\in$ $\mathcal{R})$, which does not break the atomicity of the protocol.
	After both parties enter settlement (the {\bf Close} phase in Sec.~\ref{cross-channel and applications}), the hierarchical settlement protocol (proposed in Sec.\ref{hierarchical-channel-scheme}) is adopted during the interaction process. 
	Specifically, $\mathcal{S}$ and $\mathcal{R}$ use the same hash lock $\mathsf{h(pre)}$ ($\mathsf{pre}$ is known only by $\mathcal{S}$) to lock the exchanged information, and set a timer ($T_3$ or $T_4$) in each blockchain ($\alpha$ or $\beta$) to avoid the situation of information being deadlocked. There are two cases we need to consider.
	\textit{First,} $\mathcal{S}$ uses $\mathsf{pre}$ to unlock the information in $\beta$. Then, $\mathcal{R}$ learns $\mathsf{pre}$ and uses it to unlock the information in $\alpha$ $(x_\mathcal{S}$ $\in$ $\mathcal{R}$ $\bigwedge$ $x_\mathcal{R}$ $\in$ $\mathcal{S})$.
	\textit{Second,} if $\mathcal{S}$ does not provide $\mathsf{pre}$, $\mathcal{R}$ cannot get $\mathsf{pre}$; thus it cannot unlock the information in $\alpha$. When the timer expires, the smart contract returns the locked information $(x_\mathcal{S}$ $\in$ $\mathcal{S}$ $\bigwedge$ $x_\mathcal{R}$ $\in$ $\mathcal{R})$. Note that $T_3$ in $\alpha$ is longer than $T_4$, ensuring that a malicious $\mathcal{S}$ in $\beta$ cannot provide $\mathsf{pre}$ after $T_3$ in $\alpha$ times-out. 
 
	However, \textit{in an asynchronous network}, each account may not be able to receive or upload information within a certain time due to network latency. This implies that $\mathcal{R}$ may not be able to receive $\mathsf{pre}$ and upload it within $T3$ after $\mathcal{S}$ provides $\mathsf{pre}$ $(x_\mathcal{S}$ $\in$ $\mathcal{S}$ $\bigwedge$ $x_\mathcal{R}$ $\in$ $\mathcal{S})$. Therefore $\mathsf{Cross}$-$\mathsf{Channel}$ sets a timer $T_5$ in $\alpha$ to ensure that if $\mathcal{R}$ cannot provide $\mathsf{pre}$ within $T_3$, any honest miner who receives $\mathsf{pre}$ can help $\mathcal{R}$ to upload $\mathsf{pre}$ within $T_5$ $(x_\mathcal{S}$ $\in$ $\mathcal{R}$ $\bigwedge$ $x_\mathcal{R}$ $\in$ $\mathcal{S})$. Correspondingly, the honest miner would get rewards from the smart contract. As for the incentive mechanism, our scheme can be compared to a specific application of some blockchains such as the main chain of Ethereum, where miners can get rewards
    for their work. 
\end{proof}
	\renewcommand\arraystretch{1.3}

%TODO:实验部分需要根据协议改进部分重新测试，测试结果多次取平均  --Yihao  0920
\section{Implementation and Performance Evaluation}\label{sec:Implementation}
In this section, we present our concrete implementation of $\mathsf{Cross}$-$\mathsf{Channel}$ and test its performance. 

\subsection{Implementation}\label{Implementation}

\noindent{\bf On-chain deployment: Ethereum and smart contract. } The Ethereum Geth\footnote{https://github.com/ethereum/go-ethereum} and Solidity\footnote{https://github.com/ethereum/solidity} come from Github. We use Geth to construct a test network for $\mathsf{Cross}$-$\mathsf{Channel}$ validation, and implement the smart contract $\xi$ based on Solidity.
In order to facilitate the interactions between smart contract and Ethereum, we adopt web3.py\footnote{https://pypi.org/project/web3} to deploy and call $\xi$.
	
\noindent{\bf Off-chain deployment: zk-SNARK and $(t,n)$-VSS. } We employ the zk-SNARK algorithm in Github\footnote{https://github.com/scipr-lab/libsnark} and $(t,n)$-VSS scheme in~\cite{pedersen1991non}. For the circuit in zk-SNARK (shown in Fig.~\ref{Fig:fair-exchange-circuit2}), we adopt MIMC to implement the $\mathsf{Encryption}$ and $\mathsf{Hash}$ algorithms because MIMC is encryption-friendly and can reduce circuit complexity and computational overhead.	
%\end{enumerate}

%\noindent{\bf Test environments. } 
We test the performance of $\mathsf{Cross}$-$\mathsf{Channel}$ on a local server and AliCloud.
%Each desktop is equipped with Intel Core i5-8500@3.00 GHz*6 and 19.40 GB RAM running 64-bit Ubuntu 18.04. 
The local server is equipped with an Intel$^\circledR$ Xeon(R) Silver 4214R CPU @ 2.40 GHz * 16 and 78.6 GB RAM running 64-bit Ubuntu 20.04.2 LTS. 
In the experiment on AliCloud, we use 50 ecs.g6.2xlarge instances, with each running Ubuntu 20.04 system Intel Xeon (Cascade Lake) Platinum 8269CY processor and having 8 vCPUs of frequency 2.5/3.2 GHz and 16 GB RAM. We start 4 docker nodes in each instance to form two 100-node blockchains based on the Proof-of-Work (PoW) consensus algorithm, and the number of transactions at each blockchain reaches up to 1,000.

\subsection{Performance Evaluation}\label{Performance}
\noindent {\bf On-chain performance: smart contract.}
In this experiment, we build a 2-level channel (including a sub-channel) in a 20-node blockchain to test the basic performance of the smart contract, i.e., the execution time and gas consumption of each function. 

As shown in TABLE~\ref{tab:sc_cost}, one can see that the execution time of each function is in the millisecond level, and $\mathsf{{Upload}}$ consumes the most gas (about 345,000).
The above results are reasonable because $\mathsf{Tx}_{\text{Upload}}$ involves more uploaded data, e.g., multiple signatures and keys (the definition of $\mathsf{Tx}_{\text{Upload}}$ is shown in Sec.~\ref{general-fair-exchange}).
TABLE~\ref{tab:comparation} displays the total smart contract costs in scenarios of $N$ currency exchanges (CE), $N$ fair exchanges (FE), and $N$ encrypted information exchanges (EIE), and compares HTLC and MAD-HTLC (the two most related cross-chain schemes) with our $\mathsf{Cross\text{-}Channel}$. Note that HTLC and MAD-HTLC do not support FE and EIE, and take one on-chain exchange for each CE operation, while our $\mathsf{Cross\text{-}Channel}$ takes only one on-chain exchange for $N$ operations, benefiting from the proposed channel scheme, where $N$ can be arbitrarily large. More specifically,
%
%  TODO: 描述TABLE Ⅱ内容   --Yihao 1007
%  完成
for CE, $\mathsf{Cross\text{-}Channel}$ consumes about 1,330,000 gas to process $N$ cross-chain exchanges and HTLC (MAD-HTLC) needs to take about 420,000$\times N$ (750,000$\times N$) gas to process the same volume of operations. 
For FE, $\mathsf{Cross\text{-}Channel}$ consumes the same gas as that for currency exchange, while for EIE, $\mathsf{Cross\text{-}Channel}$ needs to consume more gas (around 2,700,000 gas) due to the adoption of the general fair exchange protocol (introduced in Sec.~\ref{general-fair-exchange}) to solve the UE problem. 

Next we test the impact of the number of sub-channels on gas consumption and throughput (TPS). Our results indicate that whenever a new sub-channel is added, 
the gas consumed by the entire protocol is increased by nearly 400,000, because both the number of $\mathsf{Tx}_{\text{Close}}$ and the amount of data in $\mathsf{Tx}_{\text{Close}}$, e.g., sub-channel receipts, are increased. The benefits obtained from this gas increase is the increased number of processed transaction receipts. For example, when the number of sub-channels is increased to $L$, $N \times L$ transaction receipts can be processed, assuming that each sub-channel can process $N$ receipts. In fact,  given a quantitative resource budget, $\mathsf{Cross\text{-}Channel}$ can handle more operations, as $N$ and $L$ can be large, compared to HTLC and its variation, implying that the system throughput with $\mathsf{Cross\text{-}Channel}$ can be significantly enhanced. 
%
%
%The execution time of each function is in the millisecond level. 

%To accurately estimate the cost of smart contract, we conduct experiments based on linear regression to simulate the gas cost of all functions.
%Compared to the average cost (21,000 gas) in Ethereum, our functions need more gas, because they need to support exchanges within hierarchical channels that are more complicated than the simple currency exchanges in flat channels, which involves more uploaded data (e.g. multiple signatures and keys). 

\begin{table}[htb]
	\centering
	\begin{threeparttable} \scriptsize
		\centering
		\begin{tabular}{cccc}
			\toprule
			Contract functions  & Execution time & Gas/ ETH/ USD\tnote{$\ast$} \\% & Applications\tnote{$\dagger$} \\
		  %	 &  & Gas/ ETH/ USD \\ % & CE/ FE/ EIE \\
			\midrule
			% Gas*15/10^9 = ETH, ETH*1086=USD
			$\mathsf{{Open}}$  & 13.663ms & 70,062/ 0.0000701/ 0.0948 \\ % & \fullcirc/ \fullcirc/ \fullcirc   \\
			%+4800
			$\mathsf{{Upload}}$ & 19.102ms & 345,146/ 0.000345/ 0.467 \\ % & \emptycirc/ \emptycirc/ \fullcirc  \\
			$\mathsf{{Appeal}}$ & 13.047ms & 57,542/ 0.0000575/ 0.0778 \\ % & \emptycirc/ \emptycirc/ \fullcirc \\
			$\mathsf{{Close}}$ & 18.905ms & 149,942/ 0.000150 / 0.203 \\ %  & \fullcirc/ \fullcirc/ \fullcirc \\
			$\mathsf{{Recover}}$ & 10.883ms & 28,219/ 0.0000282/ 0.0382 \\ % &  \emptycirc/ \emptycirc/ \fullcirc  \\
			$\mathsf{{Lock}}$ & 14.843ms & 146,300/ 0.000146 / 0.198 \\ % & \fullcirc/ \fullcirc/  \fullcirc \\
			$\mathsf{{Update}}$ & 14.062ms & 79,121/ 0.0000791/ 0.107 \\ %  & \fullcirc/ \fullcirc/ \emptycirc  \\
			$\mathsf{{Update\text{-}EIE}}$ & 14.578ms  & 108,791/ 0.000109/ 0.147 \\ % & \emptycirc/ \emptycirc/ \fullcirc \\
			%	$\mathsf{{Settlement}}$  & 82453 & 0.0000825 & 0.0556 \\
			\bottomrule
		\end{tabular}
		\begin{tablenotes}
			\footnotesize
			\item[$\ast$] Gasprice = 1 Gwei, 1 Ether = $10^9$ Gwei, and 1 Ether = 1353 USD.
		%	\item[$\dagger$] We use \fullcirc \ when an application (currency exhcnage (CE), fair exchange (FE) or encrypted infomation exchange (EIE)) includes this function, use \emptycirc \ when it does not.
		\end{tablenotes}
	\end{threeparttable}
	\caption{Smart-contract experiments of $\mathsf{Cross}$-$\mathsf{Channel}$}\label{tab:sc_cost}
	%	\vspace{-0.5cm}
\end{table}
\begin{table}[htb]
	\centering
	\begin{threeparttable}\scriptsize
		\centering
		\begin{tabular}{cccccc}	
			\toprule
			&       Cross-Channel     &   Cross-Channel  & HTLC\tnote{$\S$} & MAD-HTLC\tnote{$\S$} \\
				&	     CE \& FE     &  EIE   &       CE           &  CE \\
			\midrule
			% Gas*15/10^9 = ETH, ETH*1086=USD
		%	Gas\tnote{$\dagger$}    & 10 &  14 & 4 & 6   \\
			%+4800			
		%	\cline{2-6}
			Gas   & 1,330,858   & 2,701,308 & 429,532$\times N$ & 758,095$\times N$   \\
		%	\specialrule{0em}{1pt}{1pt}
		%	\hline
		%	\specialrule{0em}{1pt}{1pt}
		%	Pessimistic\tnote{$\dagger$} & 12  &  18 & -- & 6 \\	
		%	\cline{2-6}
		%	Gas & 1,021,850 & 1,359,821 & --  & 761,965 \\
	%		\specialrule{0em}{1pt}{1pt}
	%		\hline
    %		\specialrule{0em}{1pt}{1pt}
	%		Exchange  & $N$ & $N$ & 1 & 1  \\
			%Transactions & processed & \multirow{2}{*}{$N*L$}  & \multirow{2}{*}{$N*L$} & \multirow{2}{*}{1} & \multirow{2}{*}{1}  \\
			%processed &  &  & &  &   \\
			\bottomrule
		\end{tabular}
		\begin{tablenotes}
			\footnotesize
			%\item[$\dagger$] Optimistic case refers to the protocol works ideally and pessimistic case takes into account the presence of malicious users (e.g. when some participants cannot send transactions in asynchronous network).
			%\item[$\ddagger$] N indicates the number of transactions successfully processed in the channel. L represents the number of sub-channels, and we assume that each channel successfully executes the same number of transactions.
			\item[$\S$] The results are calculated based on the data provided in~\cite{tsabary2021mad}.
			%\item[$\P$]  1 Gas = 1 Gwei, and 1 ether=2223 USD.
		\end{tablenotes}
	\end{threeparttable}
	\caption{The comparison with other cross-chain protocols}\label{tab:comparation}
\end{table}
\begin{figure*}[htb]
	\subfigure[]{\includegraphics[width=0.24\textwidth]{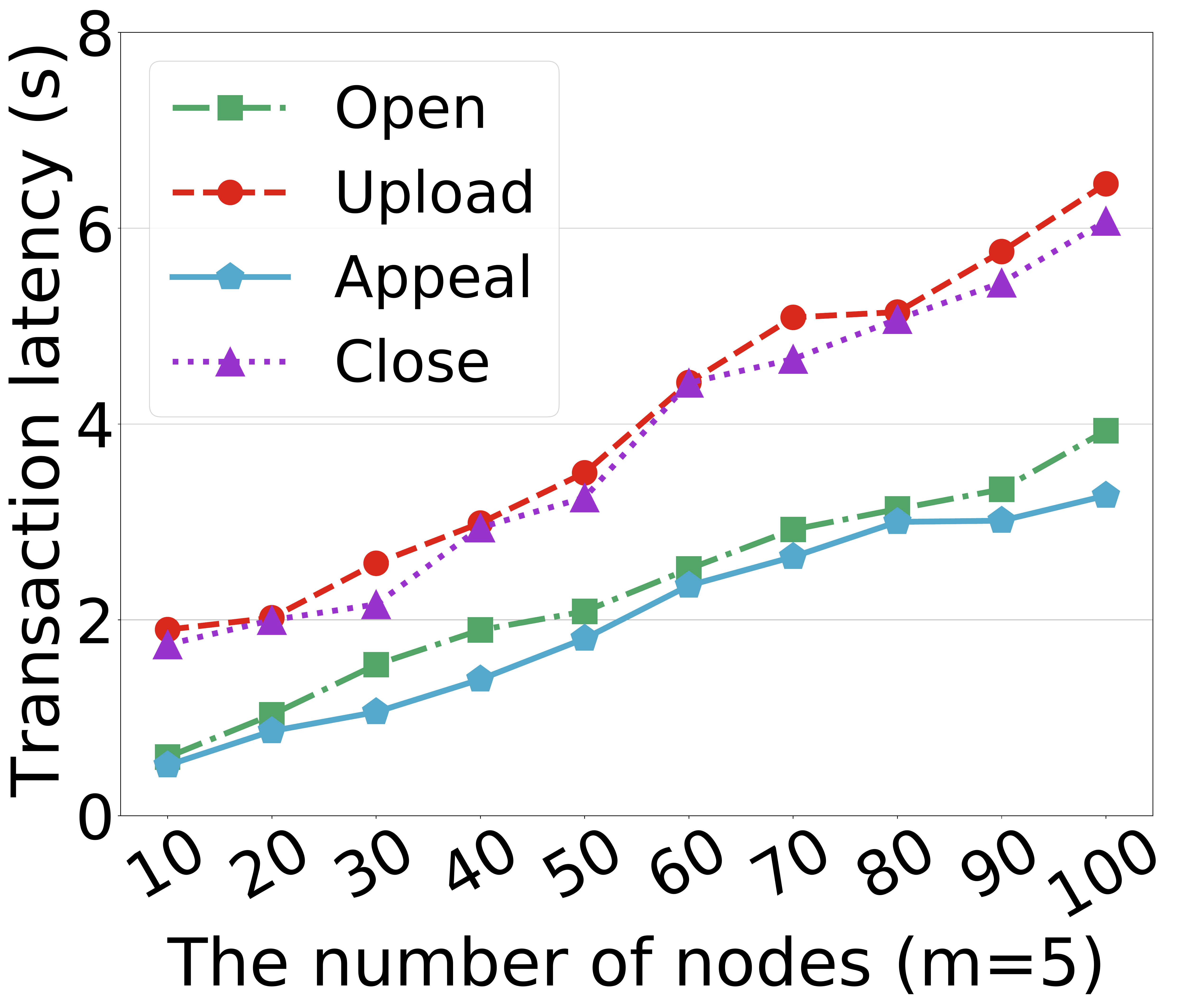}}
	\subfigure[]{\includegraphics[width=0.24\textwidth]{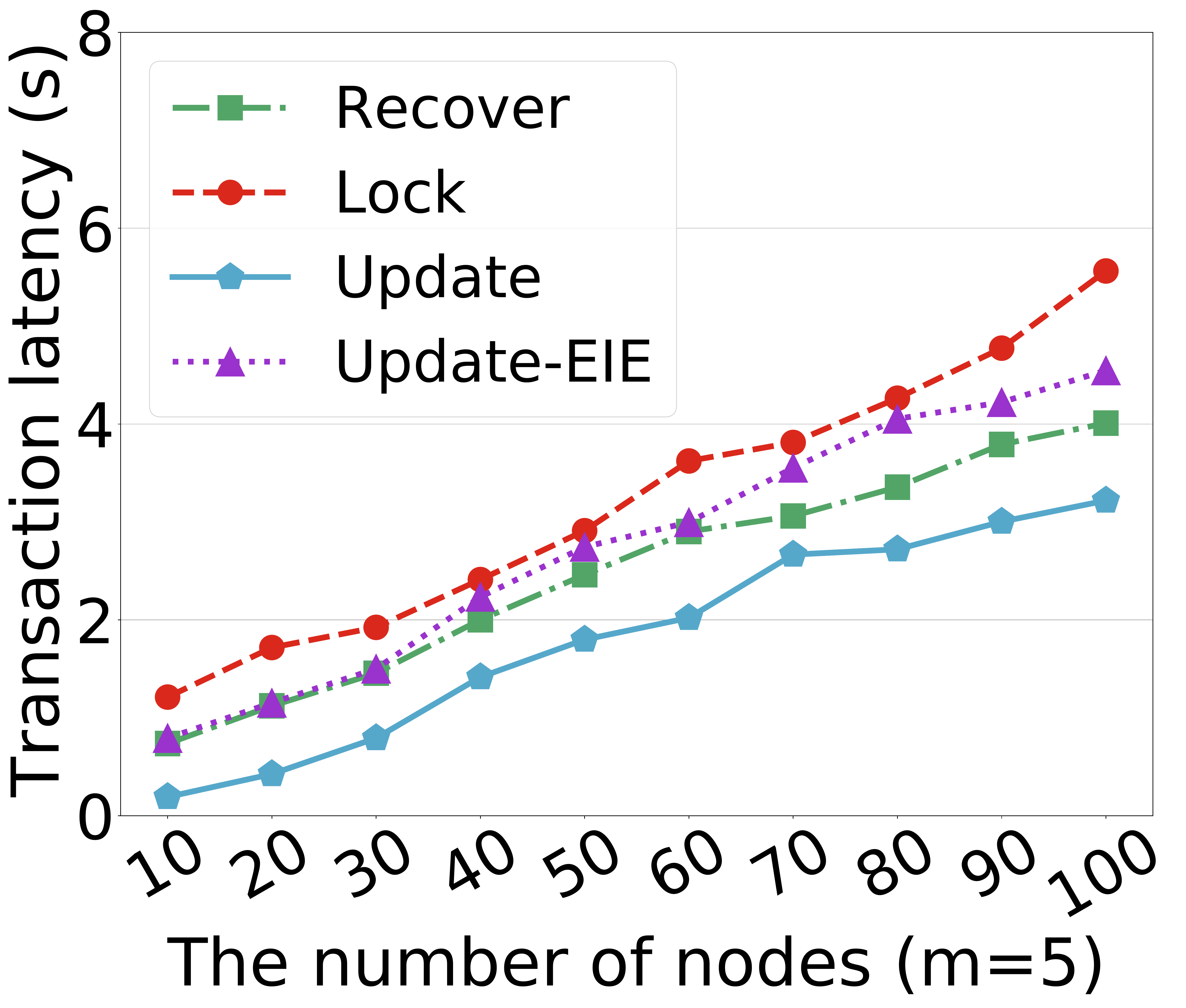}}
    \subfigure[]{\includegraphics[width=0.24\textwidth]{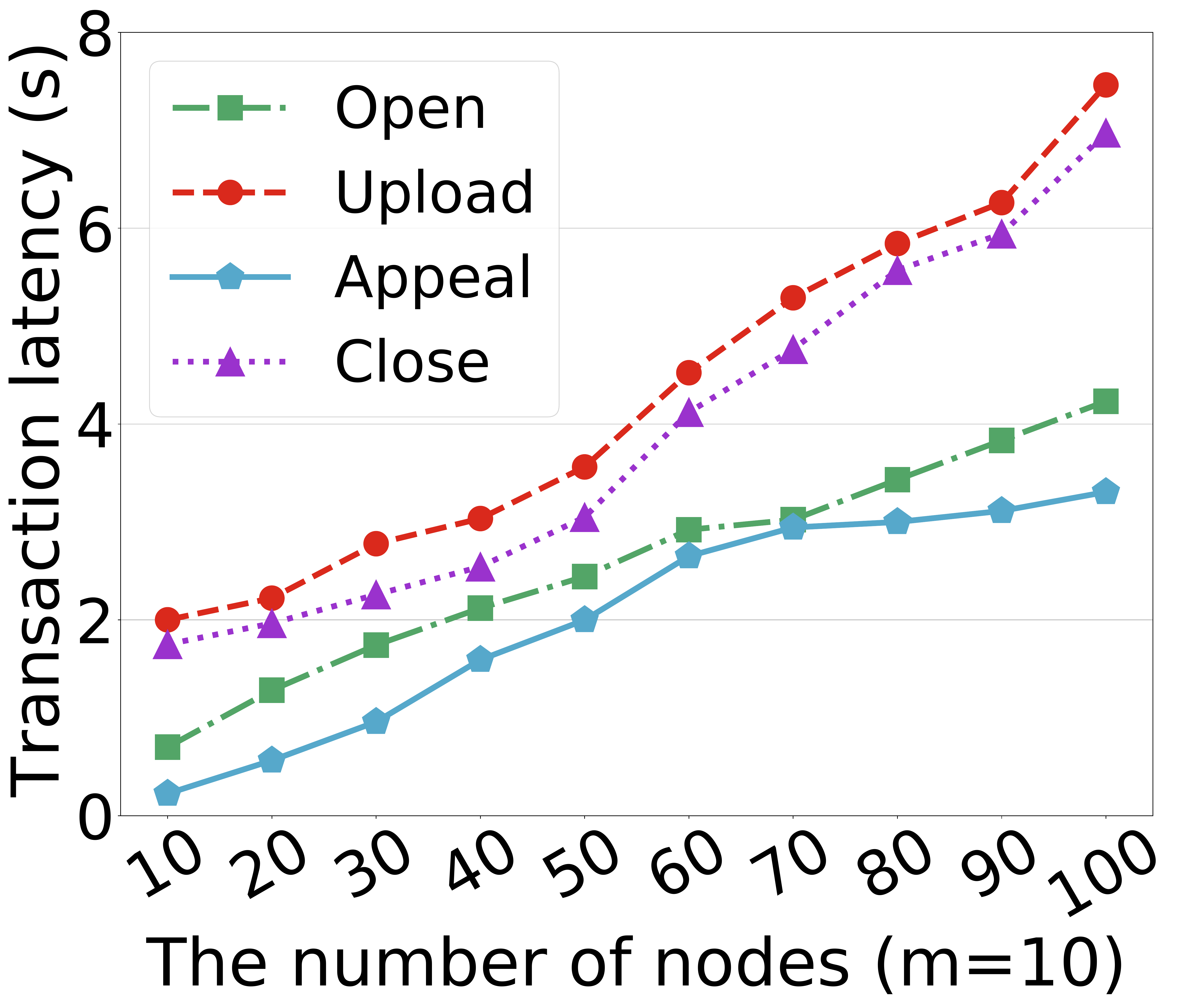}}
	\subfigure[]{\includegraphics[width=0.24\textwidth]{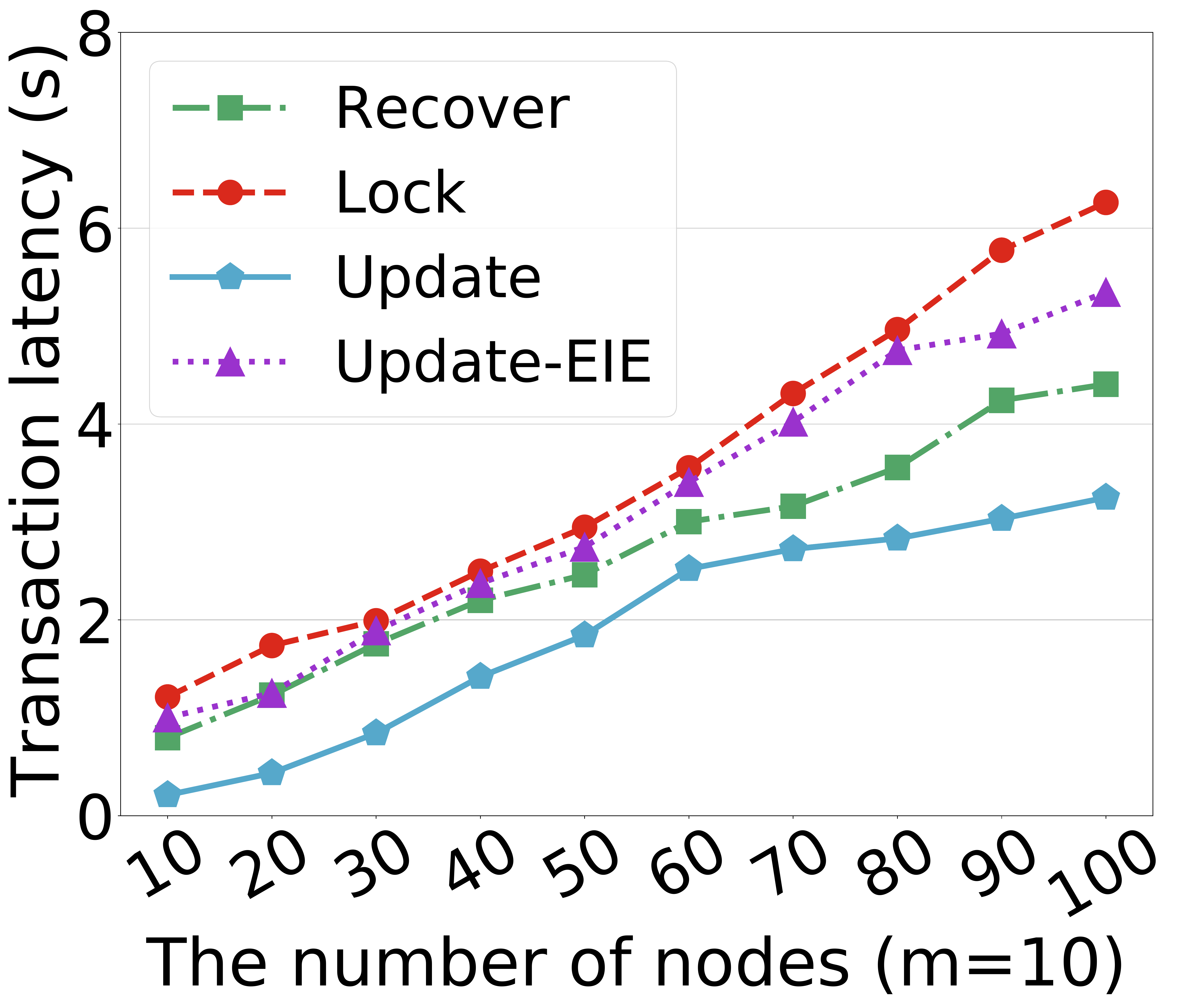}}
	\subfigure[]{\includegraphics[width=0.24\textwidth]{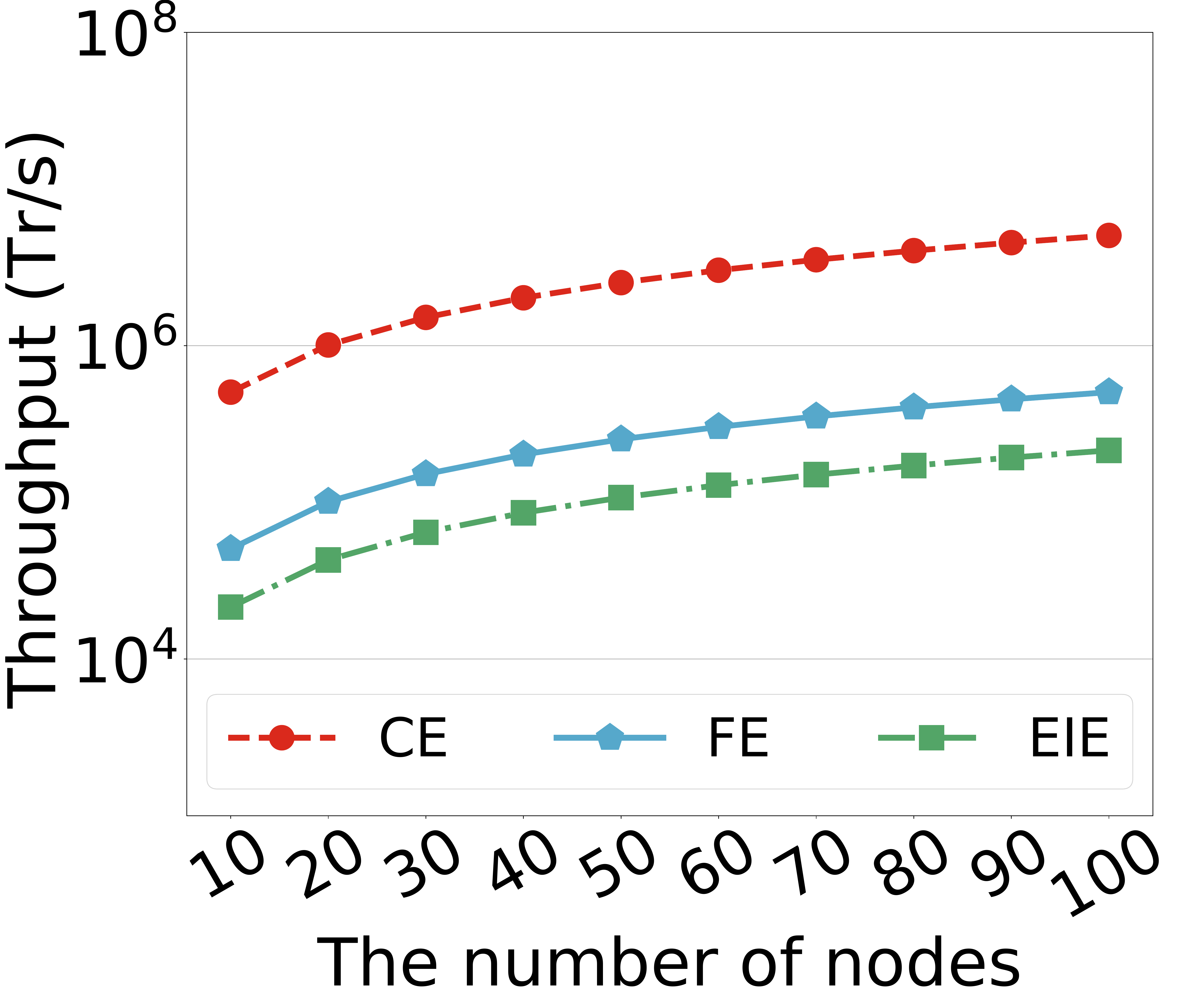}}
	%\subfigure[]{\includegraphics[width=0.32\textwidth]{figure/CC_node_TPS4.pdf}}
	\subfigure[]{\includegraphics[width=0.24\textwidth]{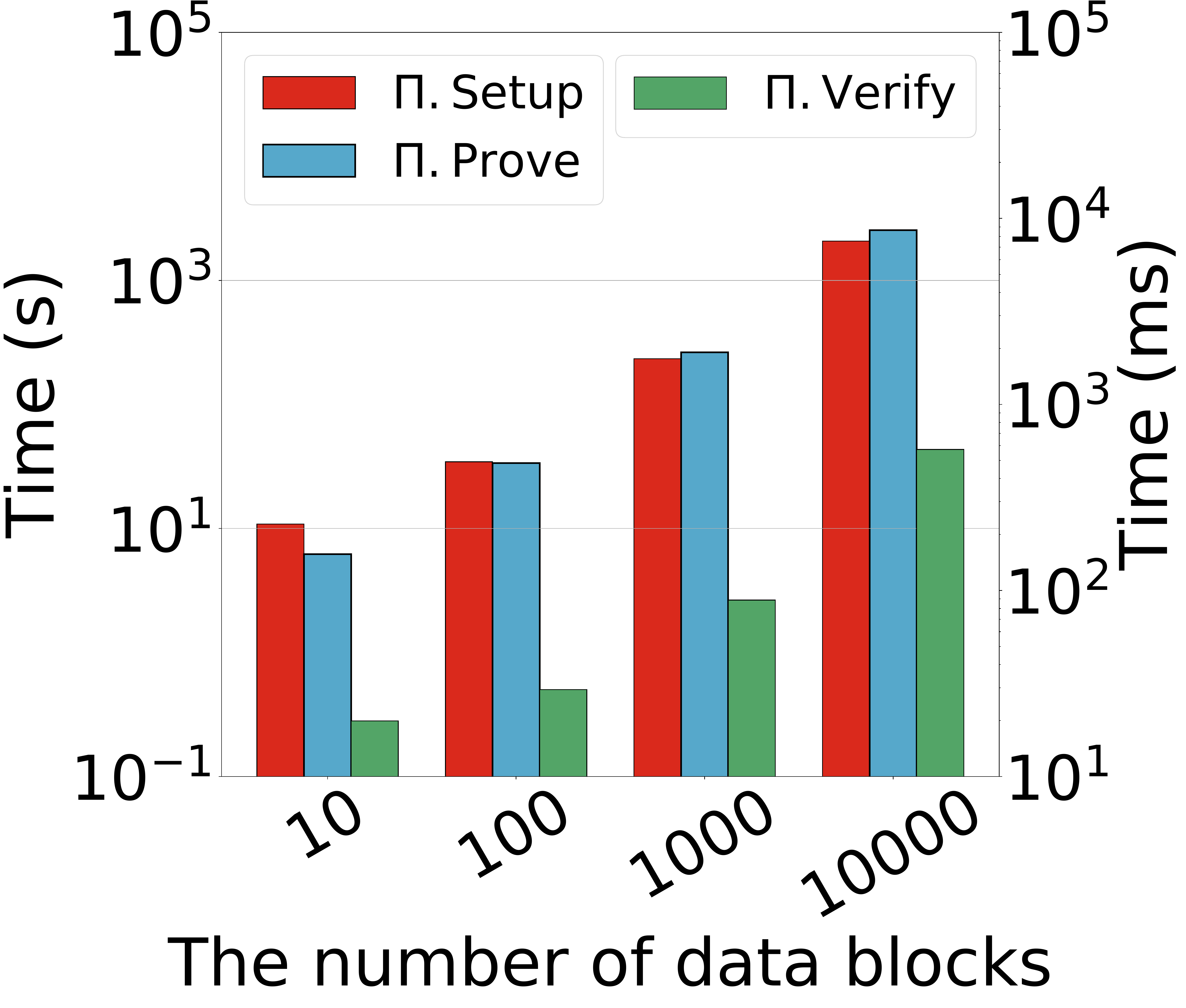}}
	\subfigure[]{\includegraphics[width=0.24\textwidth]{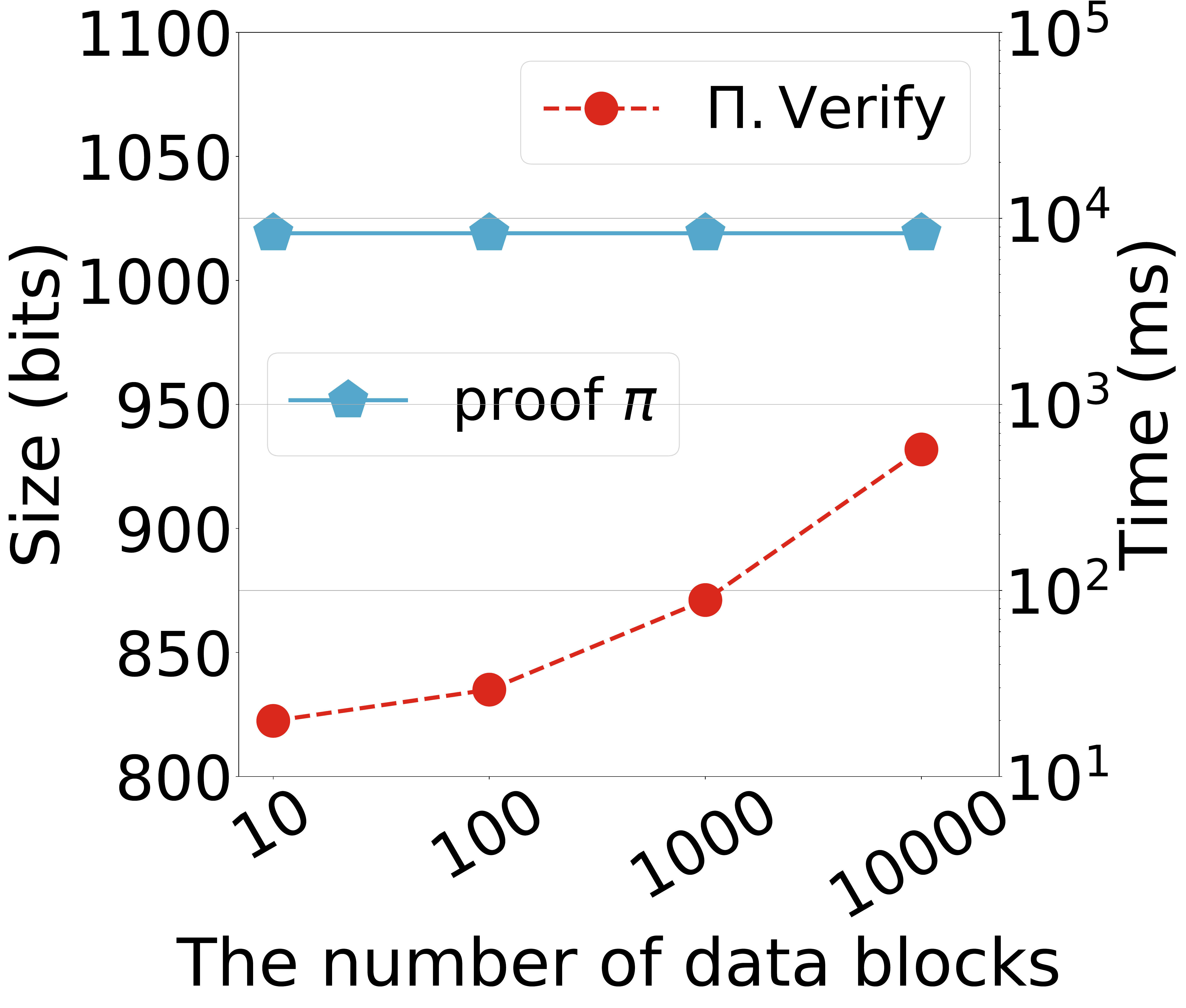}}
	\subfigure[]{\includegraphics[width=0.24\textwidth]{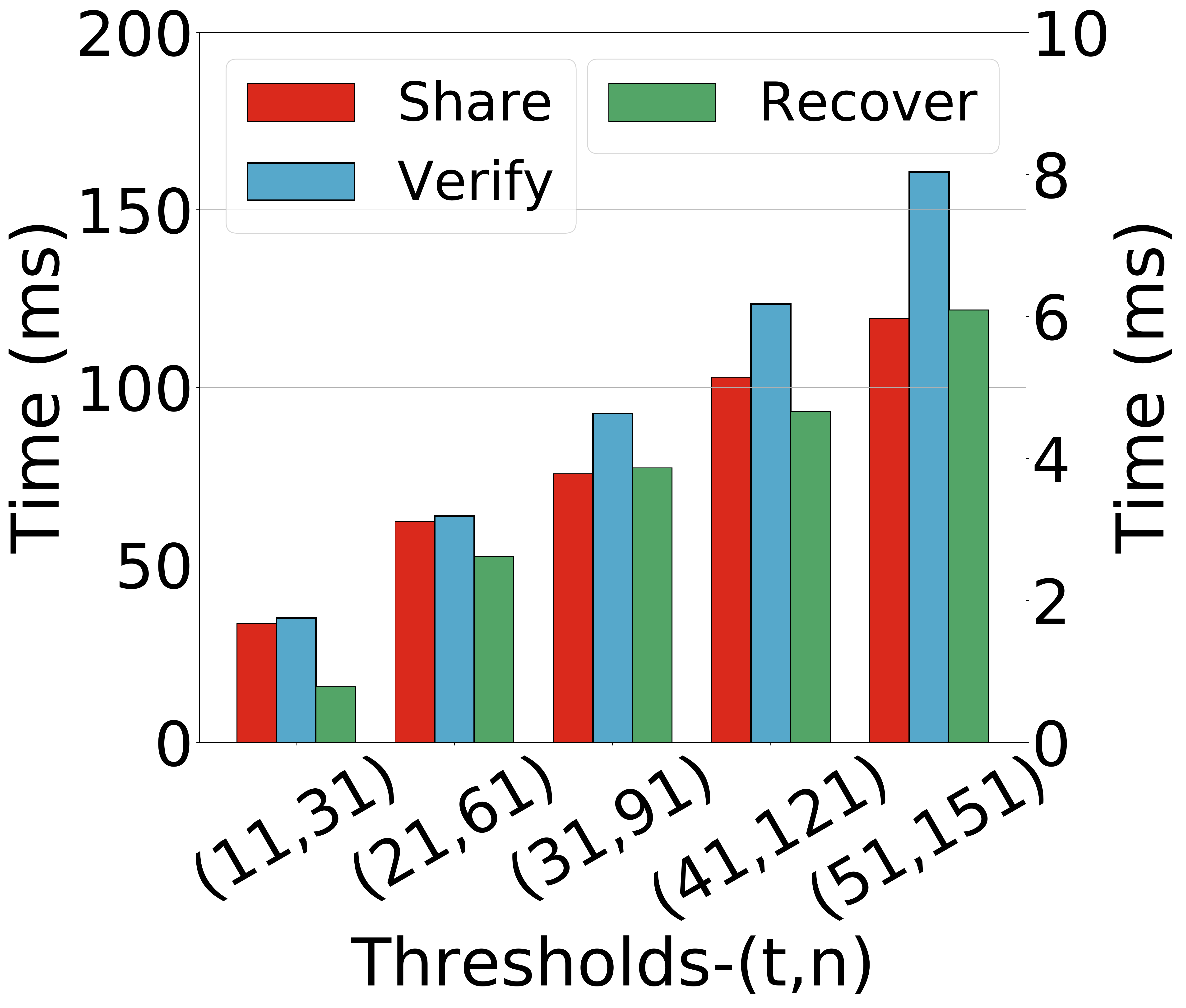}}
	\caption{Transaction delay tests (a) (b) (c) (d), TPS experiments (e), and the performance of zk-SNARK (f) (g) and verifiable secret sharing scheme (h)
	.} 
	\label{fig:test}
\end{figure*}

\noindent {\bf On-chain performance: transaction delay.} 
We simulate an EIE scenario to test the latency of each transaction. Specifically, we first build a $N_{node}$-blokchain network on AliCloud, where $N_{node}$ is the number of nodes in the blockchain. Then we let each node opens $m$ channels and each channel has two levels (i.e., including one sub-channel); each channel is open for about 100 seconds in average, during which users are allowed to interact (EIE) within the channel. The transmission rate is set to be roughly 390 MB/s.
Based on the above experimental setup, we test the transaction delay by changing $N_{node}$ and $m$, where $N_{node}$ varies from 10 to 100 and $m \in \{5,10\}$. 
%Let $N_{node}$ be the number of nodes in the blockchain.
%Assume that each node opens $m$ channels. In this experiment, the $N_{node}$-blokchain is built on AliCloud. We test the transaction delay by changing $N_{node}$ and $m$, where $N_{node}$ varies from 10 to 100 and $m \in \{5,10\}$. 
Note that the transaction delay refers to the time interval from when a transaction is sent to the blockchain until the corresponding block is confirmed by the miners. %, and the throughput refers to the number of transaction receipts that the channel can process per second.

The transaction delays are reported in Fig.~\ref{fig:test} (a) (b) (c) (d). One can see that when the number of nodes rises from 10 to 100, the transaction delay increases. The reason for this trend lies in that the more nodes in the network, the longer time the broadcast and consensus of transactions consume.
Besides that, the number of channels created in the network would also affect the transaction delay. For example, in a 100-node network, when $m = 5$, which means that 500 channels in total are constructed between nodes, the transaction latency is about 3.0-6.0 seconds (Fig.~\ref{fig:test} (a) (b)). When $m = 10$, the latency of various transactions for constructing 1,000 channels grows to about 3.2-7.4 seconds (Fig.~\ref{fig:test} (c) (d)).
The reason for this trend is that when the number of created channels increases, the number of transactions waiting in the queue increases, thus increasing the transaction delay.

%为了说明可扩展假设一个动态变化的场景：
%网络中通道总数不变（与节点成正比），但是其中有10个通道会开启，10个通道关闭。
\noindent {\bf On-chain performance: throughput (TPS) and scalability.}
In this experiment, we simulate a dynamic equilibrium state with a fixed number of channels (one-level), and discuss the throughput and scalability in three scenarios, i.e., CE, FE and EIE. %in two cases (remarked as Case-1 and Case-2).
The unit for throughput is TPS, which refers to the number of transaction receipts our scheme can process per second. 
Specifically, the numbers of channels opened and closed are dynamic variables, denoted as $v_{1}$ and $v_{2}$, respectively. We set $v_{1} = v_{2} = 10$ and $m = 5$, ensuring that in a $N_{node}$-blockchain, the total number of channels remains unchanged ($5 \times N_{node}$), but the number of newly opened channels and that of closed ones are both set to 10, maintaining a dynamic balance. 
%For precise testing, we set the interaction time of channels participating in dynamic changes as 10 seconds.
Besides that, we set the transmission rate of a channel to be roughly 390 MB/s and the test time lasts 100 seconds.
The transaction receipt is about 130 Bytes in CE (the definition of transaction receipt is described in Sec. \ref{hierarchical-channel-scheme}), and around 1.3 KB in FE and EIE (adding encrypted data blocks and information related to zero-knowledge proofs). Based on the above data, one can obtain that a channel can send $3 \times 10^6$ transaction receipts per second ($Tr$/s) in CE, and $3 \times 10^5$ $Tr$/s in FE and EIE. 
Note that, the values of variables $v_{1}, v_{2}, m$ do not affect the trend of the experimental results; thus we make them fixed.

Fig.~\ref{fig:test} (e) demonstrates that the TPS of $\mathsf{Cross\text{-}Channel}$ is linear to the network size, showing good scalability. This implies that the more nodes (channels) in the network, the higher the system throughput.
In a $N_{node}$-blockchain, the number of channels in the network is proportional to that of the nodes, and channels can process transaction receipts in parallel; thus the overall transaction processing rate of the system is nearly $\mathcal{O}(N_{node})$,
and the time to process the above transaction receipts grows at rate $\mathcal{O}(1)$.

%However, in Fig.~\ref{fig:test} (f), when the dynamic variable is proportional to the number of channel (proportional constant is 1), the system throughput decreases with the number of channels.
%So, based on the the above experimental results, the Theorem~\ref{theorem-1} can be proved correct.

\noindent {\bf Off-chain performance: zk-SNARK and VSS.}
In this experiment, we first test the effect of the size of the exchanged encrypted objects on the performance of zk-SNARK. The encrypted object includes multiple data blocks, and each block has 100-bit data.
%The performance of zk-SNARK in our local server is reported in Fig.~\ref{fig:test} (e), which plays an important role in our general fair exchange protocol $\Theta$. 
%
%In our experiment, the encrypted object includes multiple data blocks, and each block has 100-bit data.
%
As shown in Fig.~\ref{fig:test} (e), one can see that as the number of data blocks goes from 10 to 10,000, the time consumption of $\Pi.\mathsf{Setup}$ and that of $\Pi.\mathsf{Prove}$ gradually increase. This does not cost extra on-chain resources, because zk-SNARK is run off-chain and is done before running $\mathsf{Cross}$-$\mathsf{Channel}$.
Additionally, as shown in Fig.~\ref{fig:test} (g), 
zk-SNARK demonstrates great succinctness. 
The proof size is kept at 1,019 bits and the running time of $\Pi.\mathsf{Verify}$ is at the millisecond level. Of course, one can further combine with other schemes, e.g. ZKCPlus~\cite{li2021zkcplus}, to optimize the performance of zero-knowledge proof based on specific scenarios.

Then, we test the three steps ($\mathsf{Share}$, $\mathsf{Verify}$ and $\mathsf{Recover}$) at different thresholds and set $(t,n)$ to be (11, 31), (21, 61), (31, 91), (41, 121), and (51, 151), which can effectively solve the Byzantine fault (details shown in Theorem~\ref{theorem-3}). 
%follow the PoW consensus requirement of  $t/n \geq $51\% to effectively defend against the 51\% attack~\cite{saad2020exploring}. 
The time consumption of $(t,n)$-VSS is illustrated in Fig.~\ref{fig:test} (h), and one can get that the time for each step increases steadily as the threshold increases but it remains at the millisecond level. 

	\section{Conclusion}\label{sec:Conclusion}
In this paper, we propose $\mathsf{Cross}$-$\mathsf{Channel}$, a scalable channel that supports cross-chain services with high throughput.
Specifically, 
we design a new hierarchical channel structure and propose a general fair exchange protocol to respectively solve the Unsettled Amount Congestion problem and the Unfair Exchange problem.
%
%a general fair exchange protocol solves the Unfair Exchange problem and make our scheme
%applicable to various cross-chain interactions.
Additionally, we design a hierarchical settlement protocol based on HTLC and incentive mechanisms, which can help $\mathsf{Cross}$-$\mathsf{Channel}$ to ensure the correctness of the settlement and enhance the atomicity of the cross-chain interactions in asynchronous networks. %and incentive manipulation attacks 
%into account, increasing its security.  
%Finally, we simulate two 100-node blockchains in Aliyun. Experimental results show that the transaction processed of $\mathsf{Cross\text{-}Channel}$ is at least $N \times$ ($N$ is the number of transactions in channels) higher compared with HTLC protocol, validating the effectiveness and scalability of our scheme. 
Finally, we implement $\mathsf{Cross}$-$\mathsf{Channel}$ in two 100-node blockchains on AliCloud, and conduct a test with up to 1,000 transactions. 
Compared with the state-of-the-art cross-chain protocols, $\mathsf{Cross}$-$\mathsf{Channel}$ adds a small amount of on-chain resource overhead but can bring high throughput. %The experimental results prove the feasibility of our scheme.
%In our future research, 
%we will consider how to reduce the dependence on smart contracts so that our scheme can be applied to all blockchain platforms.
%In addition, 
%we will investigate novel schemes that can realize decentralized cross-chain asset transfer without relying on third parties.
In our future research, we will extend $\mathsf{Cross}$-$\mathsf{Channel}$ to support multiparty channels, and consider more general operations such as digital asset transfers.

	\bibliographystyle{ieeetr}
	
	\bibliography{IEEEabrv, References}
	
\end{document}